\documentclass[journal]{IEEEtran}
\usepackage{ifluatex}
\usepackage{cite}
\usepackage{amsmath,amssymb,amsfonts}
\usepackage{amsthm}
\usepackage{algorithmic}
\usepackage{graphicx}
\usepackage{subfigure}
\usepackage{textcomp}
\usepackage{xcolor}
\usepackage{url}
\usepackage{footnote}
\usepackage{algorithmic}
\usepackage{algorithm}
\usepackage{epsfig, comment, color}
\usepackage{mathrsfs}
\usepackage{multirow}
\usepackage{eqparbox}
\usepackage{epstopdf}
\usepackage{multicol}
\usepackage{bm}
\usepackage{makecell}
\usepackage{color}
\DeclareMathOperator*{\argmin}{arg\,min}
\DeclareMathOperator*{\argmax}{arg\,max}
\newtheorem{theorem}{Theorem}
\newtheorem{property}{Property}
\ifCLASSINFOpdf
\else
\fi
\begin{document}
\title{Optimal Caching for Low Latency in Distributed Coded Storage Systems}
\author{Kaiyang~Liu,~\IEEEmembership{Member,~IEEE,}
        ~Jun~Peng,~\IEEEmembership{Member,~IEEE,}
        ~Jingrong~Wang,~\IEEEmembership{Student Member,~IEEE,}
		and~Jianping~Pan,~\IEEEmembership{Senior Member,~IEEE}
		\IEEEcompsocitemizethanks{\IEEEcompsocthanksitem K. Liu is with the School of Computer Science and Engineering, Central South University, Changsha 410075, China, and also with the Department of Computer Science, University of Victoria, Victoria, BC V8W 2Y2, Canada. E-mail: liukaiyang@uvic.ca.
        \IEEEcompsocthanksitem J. Peng is with the School of Computer Science and Engineering, Central South University, Changsha 410075, China. E-mail: pengj@csu.edu.cn.
        \IEEEcompsocthanksitem J. Wang is with the Department of Electrical and Computer Engineering, University of Toronto, Toronto, ON M5S 1A1, Canada. E-mail: jr.wang@mail.utoronto.ca.
		\IEEEcompsocthanksitem J. Pan is with the Department of Computer Science, University of Victoria, Victoria, BC V8W 2Y2, Canada. E-mail: pan@uvic.ca.
}
}
\markboth{IEEE/ACM Transactions on Networking}%
{Shell \MakeLowercase{\textit{et al.}}: Bare Demo of IEEEtran.cls for IEEE Journals}
\maketitle
\begin{abstract}
Erasure codes have been widely considered a promising solution to enhance data reliability at low storage costs.
However, in modern geo-distributed storage systems, erasure codes may incur high data access latency as they require data retrieval from multiple remote storage nodes.
This hinders the extensive application of erasure codes to data-intensive applications.
This paper proposes novel caching schemes to achieve low latency in distributed coded storage systems.
Experiments based on Amazon Simple Storage Service confirm the positive correlation between the latency and the physical distance of data retrieval.
The average data access latency is used the performance metric to quantify the benefits of caching.
Assuming that the future data popularity and network latency information is available, an offline caching scheme is proposed to find the optimal caching solution.
Guided by the optimal scheme, an online caching scheme is proposed according to the measured data popularity and network latency information in real time.
Experiment results demonstrate that the online scheme can approximate the optimal scheme well with dramatically reduced computation complexity.
\end{abstract}

\begin{IEEEkeywords}
distributed storage systems, erasure codes, optimal caching
\end{IEEEkeywords}

\IEEEpeerreviewmaketitle

\section{Introduction}

\IEEEPARstart{I}n the big data era, the world has witnessed the explosive growth of data-intensive applications.
IDC predicts that the volume of global data will reach a staggering 175 Zettabytes by 2025~\cite{IDC_18}.
Modern distributed storage systems, e.g., Amazon Simple Storage Service (S3)~\cite{AmazonS3}, Google Cloud Storage~\cite{GoogleCloud}, and Microsoft Azure~\cite{Azure}, use two redundancy schemes, i.e., data replication and erasure codes, to enhance data reliability.

By creating full data copies at storage nodes near end users, data replication can reduce the data service latency with good fault tolerance performance~\cite{Liu_DataBot_19,Annamalai_Akkio_18}.
However, it suffers from high bandwidth and storage costs with the growing number of data replicas.
With erasure codes, each data item is coded into data chunks and parity chunks.
Compared with replication, erasure codes can lower the bandwidth and storage costs by an order of magnitude while with the same or better level of data reliability~\cite{Hu_SoCC17,EC_Store}.
Nevertheless, in geo-distributed storage systems, erasure codes may incur high access latency since 1) end users need to contact multiple remote storage nodes to access data~\cite{Agar_17}, and 2) the decoding process with parity chunks incurs non-negligible computation overheads.
The high latency prevents the further application of erasure codes to data-intensive applications, limiting its use to rarely-accessed archive data~\cite{RE_Store}.

As supplements to the geo-distributed storage system, major content providers, e.g., Akamai and Google, deploy frontend servers to achieve low latency~\cite{PANDO_20}.
End users issue requests to their nearest frontend servers, which have cached a pool of popular data items.
Nevertheless, data caching faces critical challenges in the distributed coded storage system.
As data items are coded into data chunks and parity chunks, the caching scheme should determine which chunks to cache for each data item.
To serve end users across the globe, spreading the coded chunks of each data item at more storage nodes can lower the latencies of geographically dispersed requests~\cite{PANDO_20}.
The latency of fetching different chunks varies as the coded chunks may be placed at geographically diverse nodes.
Since the data request latency is determined by the slowest chunk retrieval, caching more chunks may not proportionally reduce the overall latency.
Traditional caching schemes at the data item level may not enjoy the full benefits of caching~\cite{Sprout}.

In this paper, we propose optimal offline and near-optimal online caching schemes that are specifically designed for the distributed coded storage systems.
Preliminary experiment results in Sec.~\ref{subsec:caching} and~\ref{subsec:Motivation} show that a positive correlation exists between the latency and the physical distance of data retrieval over the wide area network (WAN).
For any two geographically diverse storage nodes, the latency gap of accessing the same data item keeps fairly stable.
The average data access latency is used as the performance metric to quantify the benefits of caching.
Assuming that the future data popularity and latency information is available, the proposed optimal scheme explores the ultimate performance of caching on latency reduction with theoretical guarantees.

Although theoretically sound in design, the optimal scheme faces the challenges of long running time and large computation overheads when applied to a large-scale storage system.
Guided by the optimal scheme, an online caching scheme is proposed with no assumption about future data popularity and network latencies.
Based on the measured data popularity and network latencies in real time, the caching decision is updated upon the arrival of each request, without completely overriding the existing caching decisions.
The theoretical analysis provides the worst-case performance guarantees of the online scheme.
The main contributions in this paper include:

\begin{itemize}
\item Novel optimal offline and near-optimal online caching schemes at the frontend servers are designed with performance guarantees for the distributed coded storage system.

\item The proposed caching schemes are extended and implemented to the case of storage server failure.

\item A prototype of the caching system is built based on Amazon S3. Extensive experiment results show that the online scheme can approximate the optimal scheme well with dramatically reduced computation overheads.

\end{itemize}

The rest of this paper is organized as follows.
Sec.~\ref{sec: related} summarizes the related work.
Sec.~\ref{sec:Modeling} presents the model of the distributed coded storage system and states the caching problem.
Sec.~\ref{sec: Optimal_caching} and~\ref{sec: Online_caching} provide the design of the optimal and online caching schemes, respectively.
Sec.~\ref{sec:Evaluation} evaluates the efficiency and performance of the caching schemes through extensive experiments.
Sec.~\ref{sec:conclusion} concludes this paper and lists future work.

\section{Related Work} \label{sec: related}


\textbf{Caching at the data item level.} Data caching has been considered a promising solution to achieve low latency in distributed storage systems.
Ma et al.~\cite{Ma_13} proposed a replacement scheme that cached a full copy of data items to reduce the storage and bandwidth costs while maintaining low latencies.
Liu et al.~\cite{DistCache} designed DistCache, a distributed caching scheme with provable load balancing performance for distributed storage systems.
Song et al.~\cite{Song_NSDI_20} proposed a learning-based scheme to approximate the Belady's optimal replacement algorithm with high efficiency~\cite{Belady}.
However, all these previous studies focus on caching at the data item level.
Due to the limited cache capacity, keeping a full copy of data items may not be space efficient to achieve the full benefits of caching with erasure codes.

\textbf{Low latency in coded storage systems.} Erasure codes have been extensively investigated in distributed storage systems as it can provide space-optimal data redundancy.
However, it is still an open problem to quantify the accurate service latency for coded storage systems~\cite{Sprout}.
Therefore, recent studies have attempted to analyze the latency bounds based on queuing theory~\cite{Sprout,MDSqueue,Xiang_TNSM_17,Badita_TIT_19}.
These researches are under the assumption of stable request arrival process and exponential service time distribution, which may not be applicable for a dynamic network scenario.
Moreover, prior works also focused on the design of data request scheduling schemes to achieve load balancing in coded storage systems~\cite{Hu_SoCC17,EC_Store,Aggarwal_17}.
In this way, the data access latency could be reduced by the avoidance of data access collision.
These scheduling schemes are suitable for intra-data center storage systems as the network congestion dominates the overall data access latency.

\textbf{Caching in coded storage systems.} Compared with schemes that cache the entire data items, Aggarwal et al.~\cite{Sprout} pointed out that caching partial data chunks had more scheduling flexibility.
Then, a caching scheme based on augmenting erasure codes was designed to reduce the data access latency.
Nevertheless, extra storage overheads were introduced with the augmented scheme.
Halalai et al.~\cite{Agar_17} designed Agar, a dynamic programming-based caching scheme to achieve low latency in coded storage systems.
Agar was a static policy that pre-computed a cache configuration for a certain period of time without any worst-case performance guarantees.
Rashmi et al.~\cite{ECCache_16} applied an online erasure coding scheme on data stored in the caching layer to achieve load balancing and latency reduction.
However, this online erasure coding scheme introduced extra computation overheads when handling massive data read requests.
Different from previous studies, our proposed caching schemes leverage the measured end-to-end latency to quantify the benefits of caching.
Furthermore, the proposed schemes only cache data chunks rather than parity chunks to avoid the decoding overheads of data read requests.

\begin{figure}[!t]
\centerline{\includegraphics[width=3.3in]{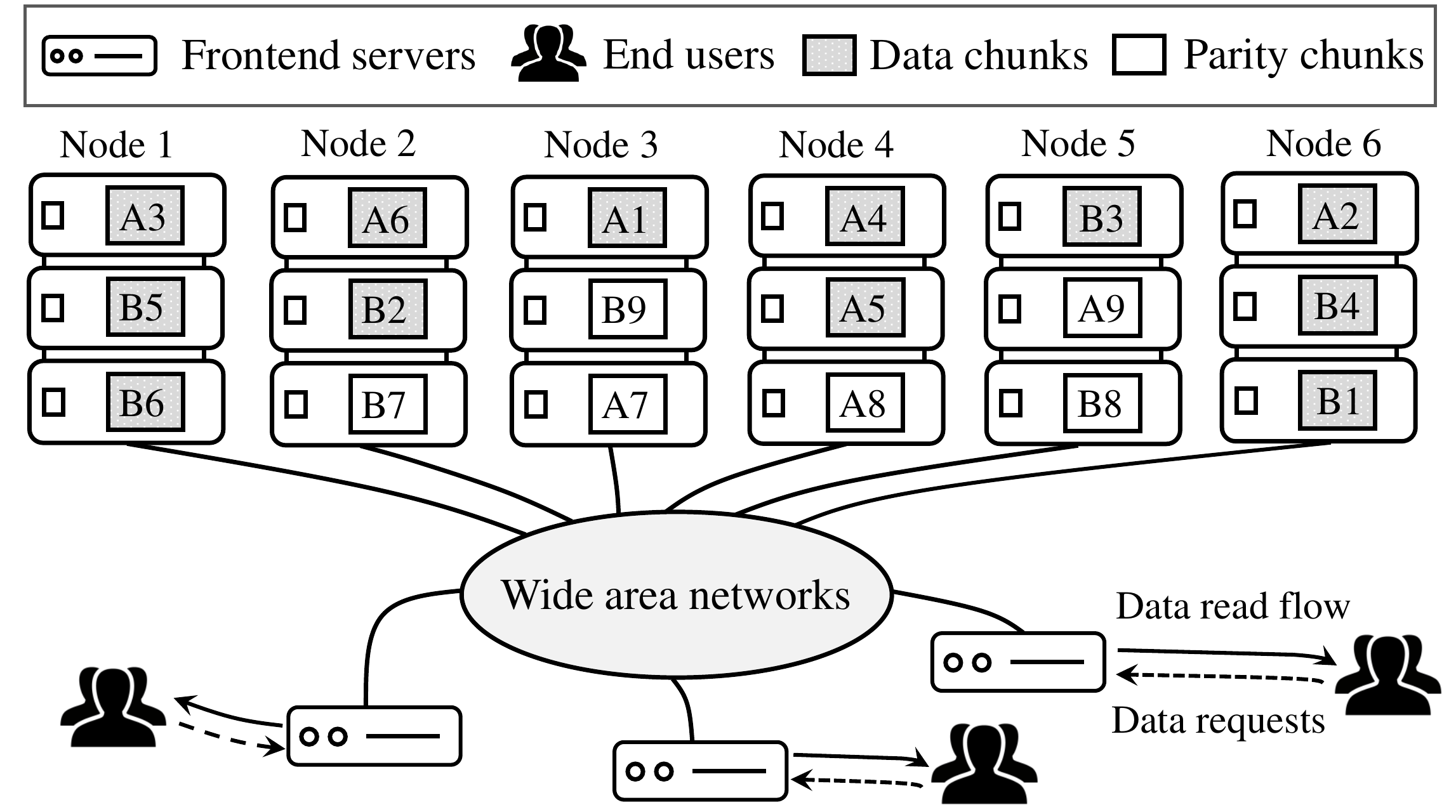}}
\caption{An illustration of the distributed coded storage system with caching services is shown. Data items A and B are coded into $K = 6$ data chunks and $R = 3$ parity chunks.}
\label{fig:System_Model}
\end{figure}

\section{System Model and Problem Statement} \label{sec:Modeling}

This section presents the architecture of the geo-distributed storage system with erasure codes and discusses how to reduce the access latency of data requests with caching services.

\subsection{Geo-distributed Storage System and Erasure Codes} \label{subsec:system}

\begin{table*}[!t] \scriptsize
\caption{The deployment of storage nodes over six Amazon Web Services (AWS) regions and the average data access latency (in milliseconds) from remote storage nodes to end users at three different locations.}
\begin{center}
\begin{tabular}{|c|c|c|c|c|c|c|c|c|c|}
\hline
\multicolumn{2}{|c|}{\textbf{Storage node}} & \textbf{1} & \textbf{2} & \textbf{3} & \textbf{4} & \textbf{5} & \textbf{6} \\
\hline
\multicolumn{2}{|c|}{\textbf{Region}} & Tokyo & Ohio & Ireland & S\~{a}o Paulo & Oregon & Northern California \\
\hline
\multirow{3}*{\textbf{Average latency (ms)}} & \textbf{Victoria, CA} & 479.3 & 345.5 & 686.3 & 803.9 & 128.3 & 179.3 \\
\cline{2-8}
& \textbf{San Francisco, US} & 513.2 & 338.4 & 663.2 & 786.9 & 158.3 & 84.7 \\
\cline{2-8}
& \textbf{Toronto, CA} & 794.7 & 129.0 & 631.5 & 705.5 & 302.6 & 355.7 \\
\cline{2-8}
\hline
\end{tabular}
\label{tab:AWS_Regions}
\end{center}
\end{table*}

As shown in Fig.~\ref{fig:System_Model}, the geo-distributed storage system consists of a set of geographically dispersed storage nodes $\mathcal{N}$ (with size $N=|\mathcal{N}|$)~\footnote{The storage nodes could be data centers in practice. As shown in Fig.~\ref{fig:System_Model}, each storage node consists of multiple storage servers.}.
The set of data items stored in the system is denoted by $\mathcal{M}$ (with size $M=|\mathcal{M}|$).
Similar to Hadoop~\cite{HDFS} and Cassandra~\cite{Cassandra}, all data items are with a default block size.
To achieve the target reliability with the maximum storage efficiency, the Reed-Solomon (RS) codes are adopted as the storage scheme~\footnote{Other erasure codes, e.g., local reconstruction codes (LRC)~\cite{LRC}, can also be applied in our solution.}.
With a linear mapping process, each data item is coded into equal-sized $K$ data chunks and $R$ parity chunks.
All coded chunks are distributed among storage nodes for fault tolerance, which are denoted by
\begin{equation} \label{equ:InitialPlacement}
\left\{\begin{matrix}
    m_k \rightarrow i, k \in \{1,...,K\}, i \in \mathcal{N},
    \\
    m_r \rightarrow j, r \in \{1,...,R\}, j \in \mathcal{N},
\end{matrix}\right.
\end{equation}
which represents chunk $m_k$ and $m_r$ are placed at node $i$ and $j$, respectively~\footnote{In this paper, data replication at the remote storage nodes is not considered to achieve low storage overheads.
Our design is also applicable to the scenario of data replication. The data request is served by fetching $K$ data chunks from the nearest storage nodes.}.
Please note that the coded chunks are not placed at a single storage node since this will increase the data access latency of end users far from that node.

When the requested data is temporarily unavailable, the original data can be recovered via the decoding process from any $K$ out of $K+R$ chunks.
The decoding process with parity chunks will inherently incur considerable computation overheads to the storage system.
Generally speaking, a read request is first served by obtaining $K$ data chunks to reconstruct the original data with low overheads~\cite{Hu_SoCC17}.
The action of parity chunk retrieval for decoding is defined as {\bf degraded read}.
The degraded read is passively triggered 1) when the storage server storing the data chunks is momentarily unavailable, or 2) during the recovery of server failure.
Moreover, the data write/update process is not considered in this paper, since most storage systems are append-only where all the data are immutable. Instead, data with any updates will be considered as separate items with new timestamps~\cite{Hu_SoCC17}.

Erasure codes may incur high data access latency in the geo-distributed storage system.
The requested chunks are retrieved by accessing multiple remote storage nodes.
The high latency impedes the extensive application of erasure codes to data-intensive applications.
Therefore, it is imperative to reduce the data request latencies in the coded storage system.

\subsection{Caching at Frontend Servers for Low Latency} \label{subsec:caching}

This paper adopts caching to achieve low latency data services.
As illustrated in Fig.~\ref{fig:System_Model}, multiple frontend servers are deployed to serve geographically dispersed end users.
Each frontend server creates an in-memory caching layer to cache popular data items near end users.
Instead of interacting directly with remote storage nodes, end users retrieve data from the frontend server.
Let $C$ denote the cache capacity of the frontend server.
Due to the scarcity of memory, not all data chunks can be cached in the caching layer, i.e., $C \leq M K$.

With erasure codes, we may not need to cache all chunks of each data item to achieve the full benefits of caching.
This can be demonstrated through preliminary experiments based on Amazon S3.
As shown in Fig.~\ref{fig:System_Model}, a prototype of the coded storage system is deployed over $N = 6$ Amazon Web Services (AWS) regions, i.e., Tokyo,
Ohio, Ireland, Sao Paulo, Oregon, and Northern California.
In each AWS region, three {\tt buckets} are created.
Each {\tt bucket} represents a server for remote data storage.
The storage system is populated with $M = 1,000$ data items.
The numbers of data and parity chunks are set as $K = 6$ and $R = 3$.
The default size of all chunks is 1 MB~\cite{Agar_17}.
For each data item, the nine coded chunks are uniformly distributed among eighteen {\tt buckets} to achieve load balancing.
The coded chunks of each data item are not placed at the same server to guarantee the $R$-fault tolerance.
As noted in prior work~\cite{ECCache_16,Hu_SoCC17}, the popularity of data items follows a Zipf distribution.

Furthermore, three frontend servers are deployed near the end users at various locations, i.e., Victoria, Canada, San Francisco, United States, and Toronto, Canada.
{\tt Memcached}~\cite{Memcached} module is adopted for data caching in RAM.
The frontend server uses a thread pool to request data chunks in parallel.
For each read request, the end user needs to obtain all six data chunks from remote {\tt buckets} without caching at the frontend server.
Table~\ref{tab:AWS_Regions} shows the average data access latency from remote {\tt buckets} to end users~\footnote{Compared with the long access latency over WAN (in hundreds of milliseconds), the reconstruction latency with data chunks and the network latency from the frontend server to end users is negligible.}.
For data requests, the latency is determined by the slowest chunk retrieval among all chunks.
As shown in Fig.~\ref{fig:System_Model}, if data item B (including data chunk B1--B6) is requested from the frontend server in Victoria, the latency is about 479.3 ms as we need to fetch data chunk B5 and B6 from the farthest storage node in Tokyo.

\begin{figure}[htbp]
\centerline{\includegraphics[width=2.7in]{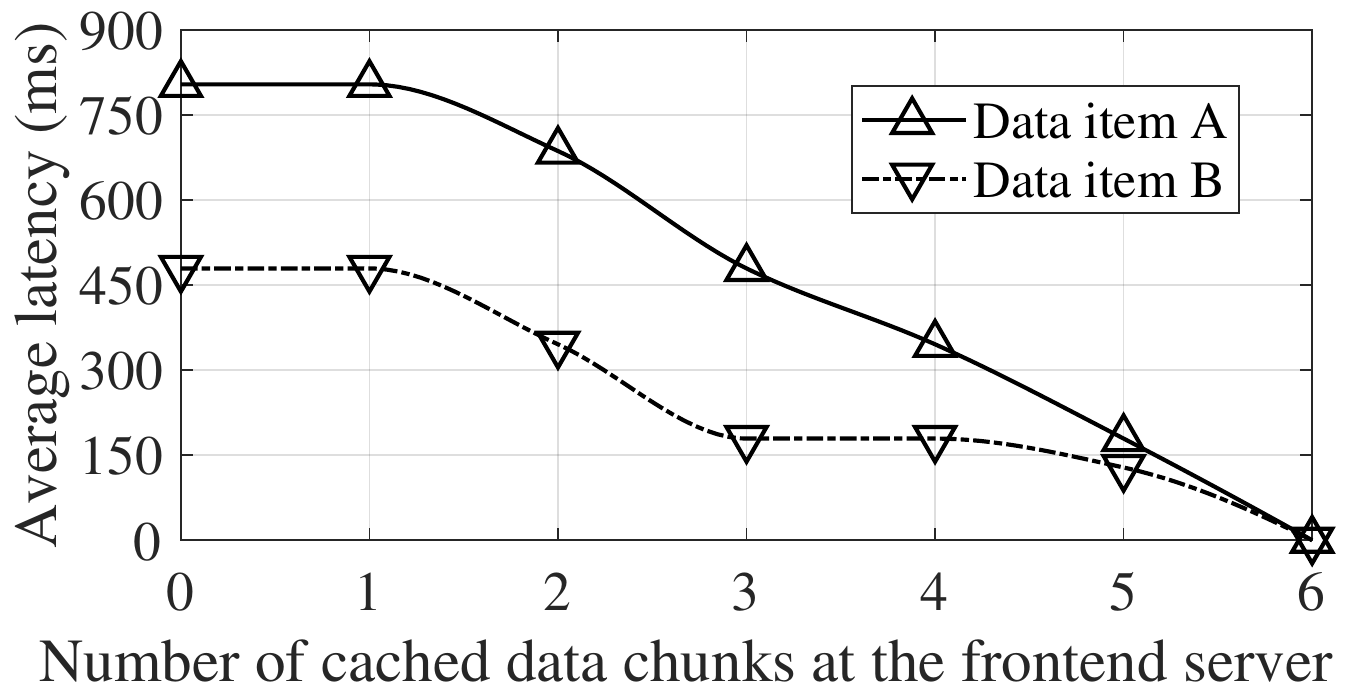}}
\caption{Experiment results show the average access latency of caching different number of data chunks on the frontend server in Victoria. The relationship between the number of cached data chunks and the reduced latency is nonlinear. The storage locations of data items A and B are shown in Fig.~\ref{fig:System_Model}.}
\label{fig:Latency_Example}
\end{figure}

\begin{figure*}[!t]
	\centering
	\subfigure[Latencies from Victoria]{
		\begin{minipage}[b]{0.31\textwidth}
			\includegraphics[width=1\textwidth]{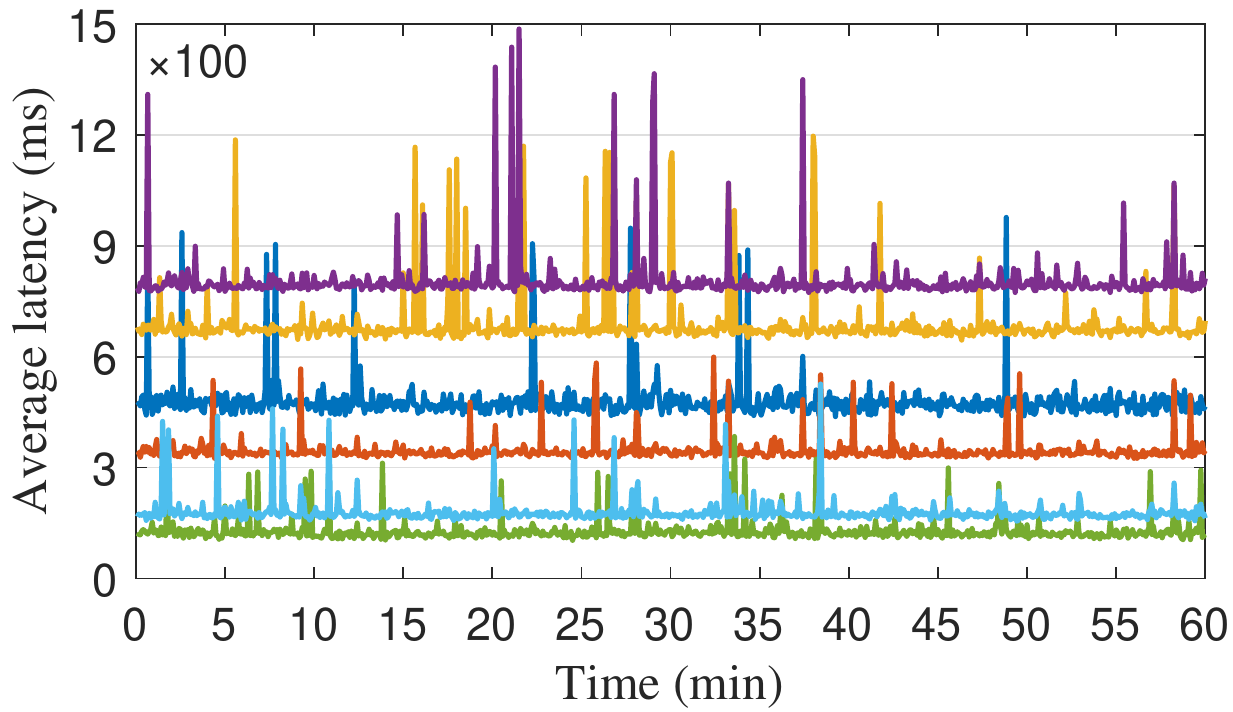}
		\end{minipage}
	}
	\subfigure[Latencies from San Francisco]{
		\begin{minipage}[b]{0.31\textwidth}
			\includegraphics[width=1\textwidth]{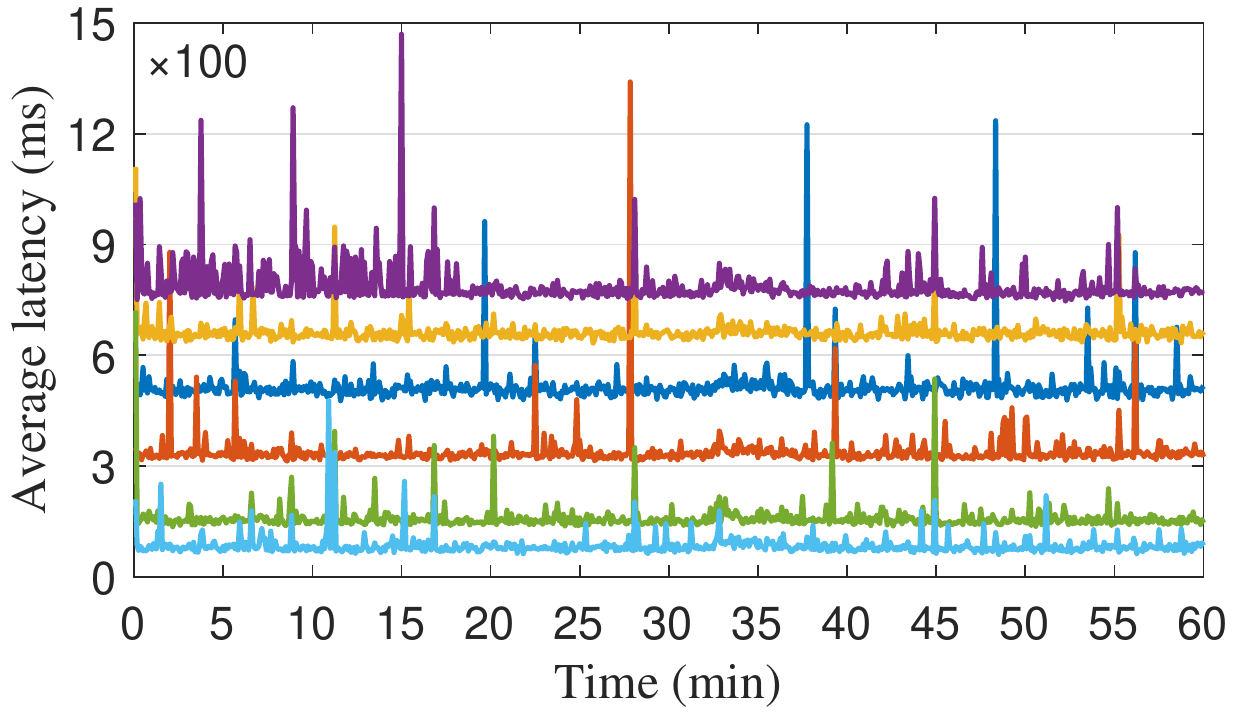}
		\end{minipage}
	}
	\subfigure[Latencies from Toronto]{
		\begin{minipage}[b]{0.31\textwidth}
			\includegraphics[width=1\textwidth]{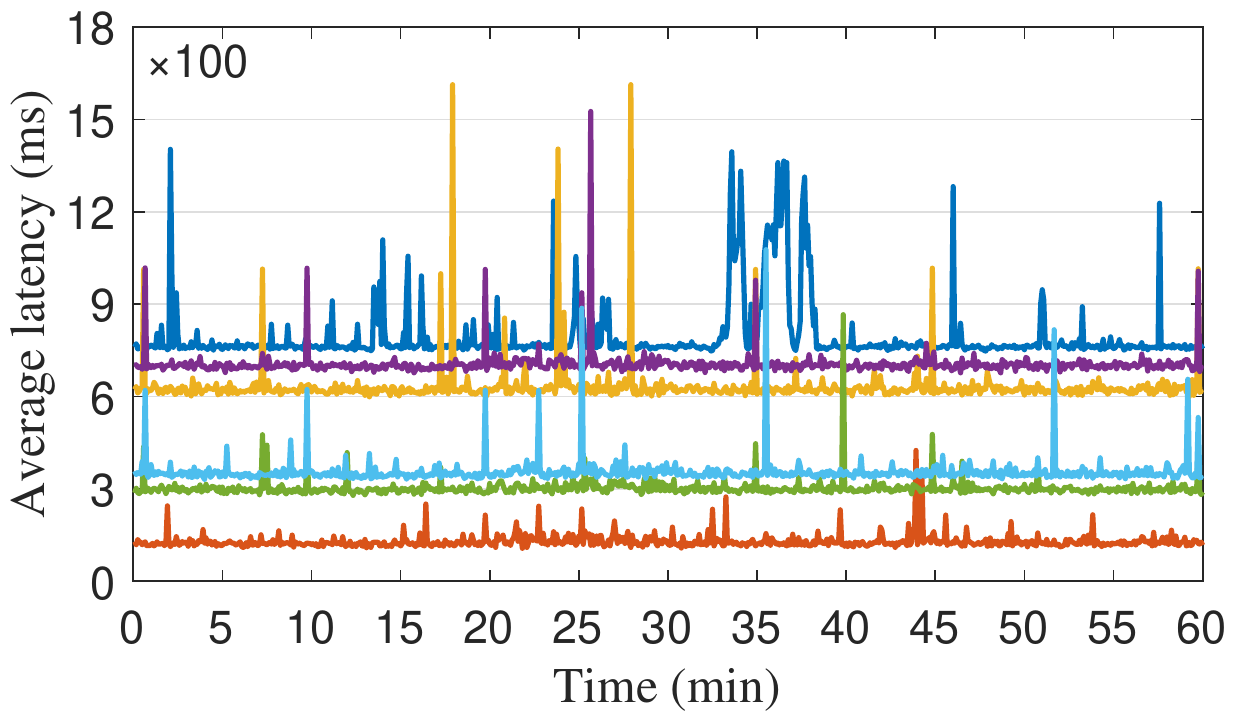}
		\end{minipage}
	}
	\subfigure[CDF of latencies from Victoria]{
		\begin{minipage}[b]{0.31\textwidth}
			\includegraphics[width=1\textwidth]{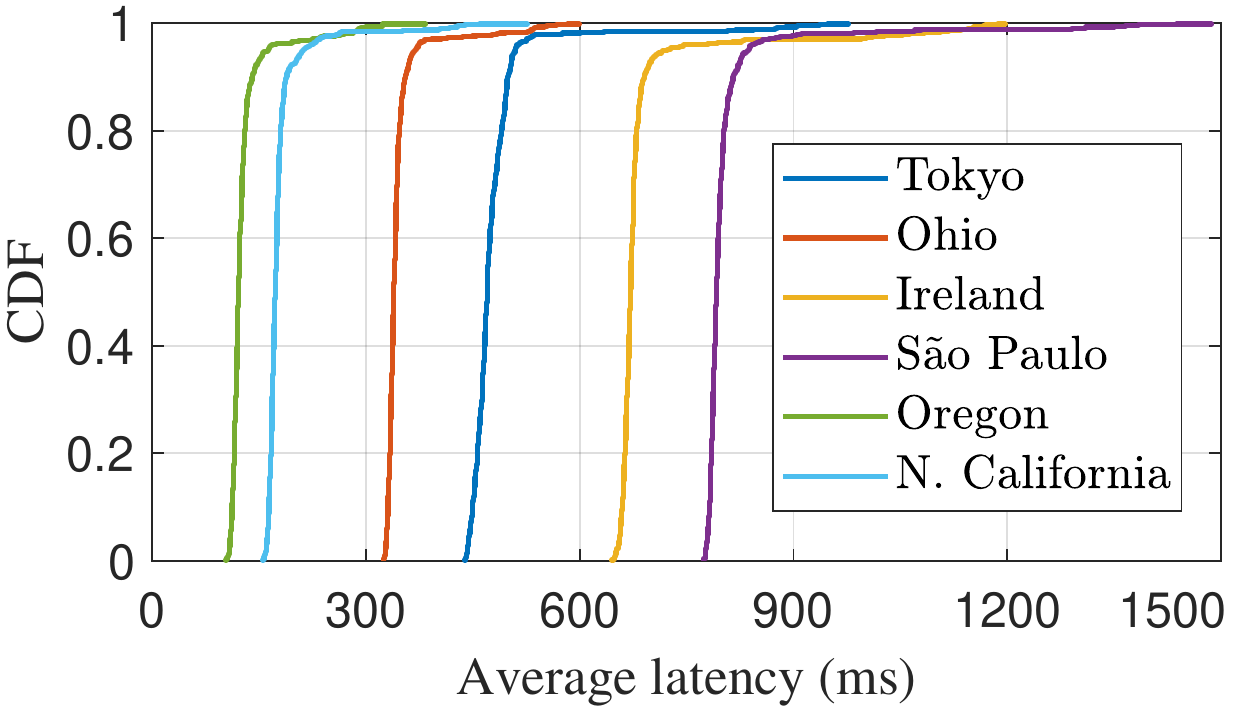}
		\end{minipage}
	}
	\subfigure[CDF of latencies from San Francisco]{
		\begin{minipage}[b]{0.31\textwidth}
			\includegraphics[width=1\textwidth]{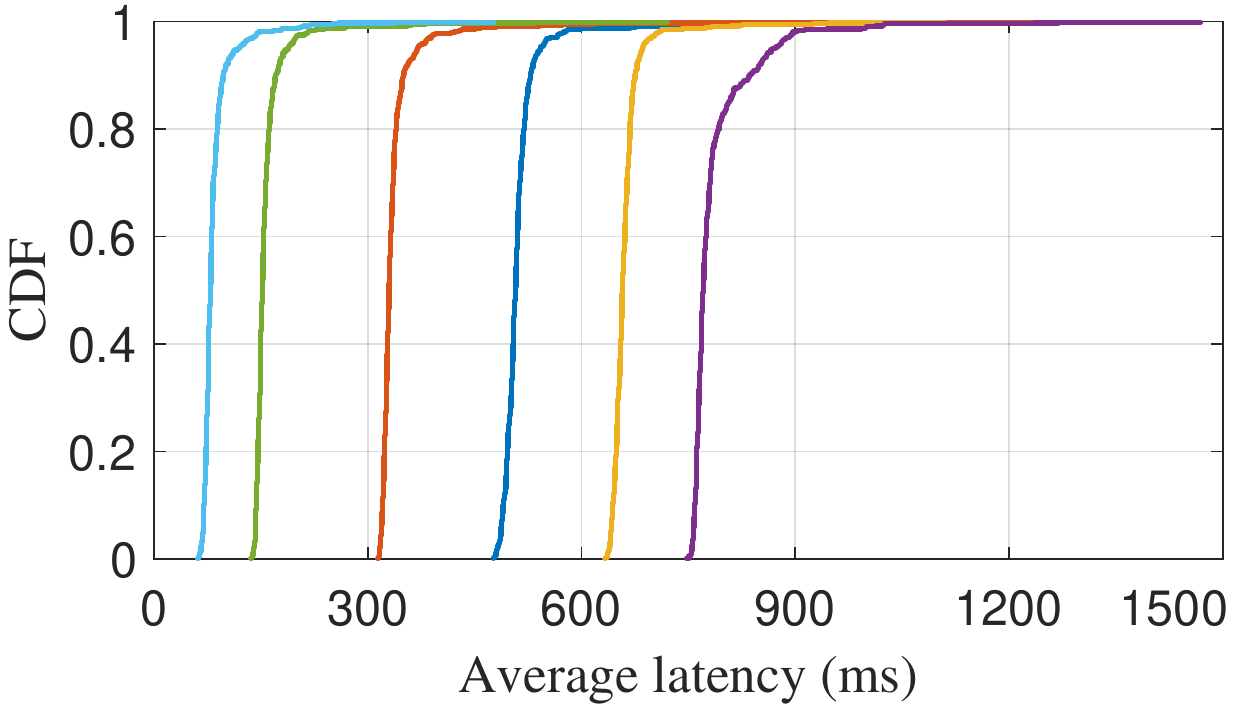}
		\end{minipage}
	}
	\subfigure[CDF of latencies from Toronto]{
		\begin{minipage}[b]{0.31\textwidth}
			\includegraphics[width=1\textwidth]{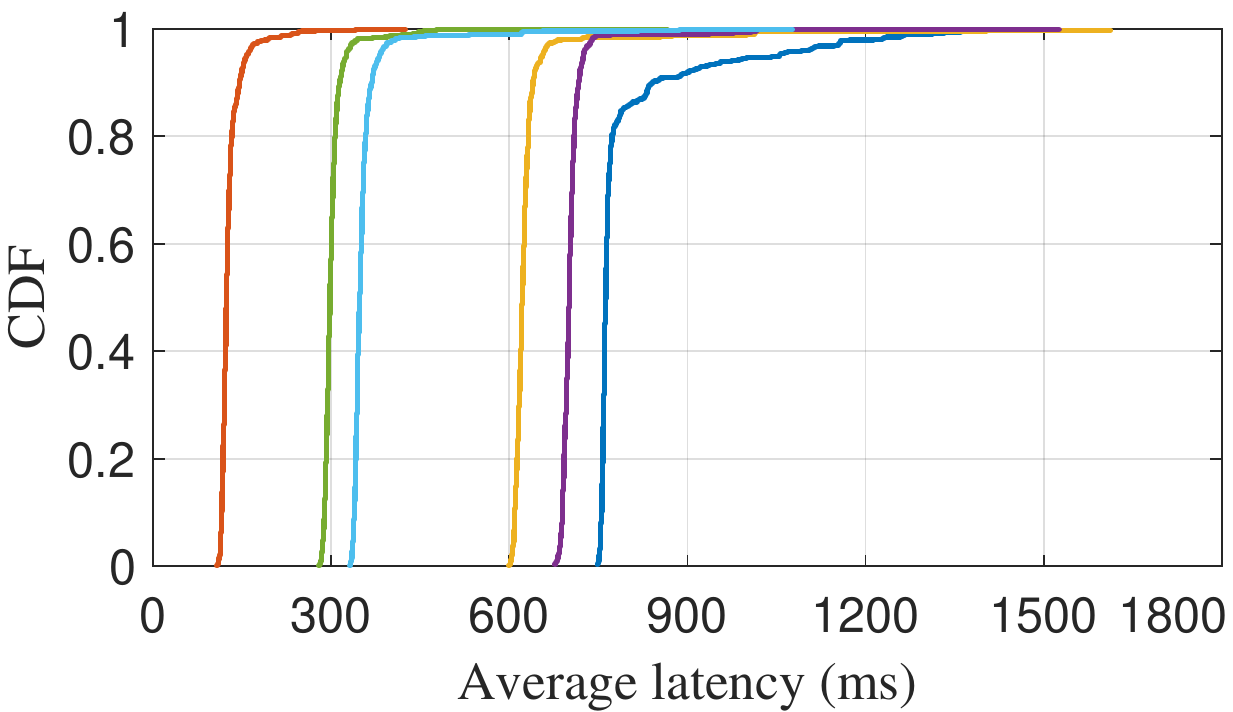}
		\end{minipage}
	}
	\caption{Experiment results show the data access latencies from storage nodes to end users over a period of one hour (from 15:00 to 16:00 on July 30, 2020).}
	\label{fig:Latency}
\end{figure*}

Then, we consider the performance of caching at the frontend server.
Fig.~\ref{fig:Latency_Example} illustrates the latency reduction performance by gradually increasing the number of cached data chunks for different data items.
In the design of the preliminary experiments, only data chunks will be cached to avoid degraded read.
The data chunk with higher access latency is cached with higher priority.
Then, the data access latency can be progressively decreased.
Fig.~\ref{fig:Latency_Example} shows the average access latency of caching different number of data chunks on the frontend server in Victoria.
Let $\varepsilon_m$ denote the number of cached data chunks for data item $m$, $\varepsilon_m \in \{0,...,K\}$, $m \in \mathcal{M}$.
We have the following two observations:

\begin{itemize}
\item The access latency function $f_m(\varepsilon_m)$ is nonlinear.

\item The storage locations of chunks may be different for various data items, e.g., data items A and B in Fig.~\ref{fig:System_Model}.
For various data items, the access latency function $f_m(\varepsilon_m)$ could also be different due to the diverse storage locations.
For instance, if three chunks are cached for data item A, the data access latency is reduced by 40.3\%.
For data item B, three cached data chunks can reduce the latency by 62.6\%.
\end{itemize}

The observations show that the latency reduction varies from data item to data item.
Traditional caching schemes at the data item level cannot achieve the full benefits of caching.

\subsection{Caching Problem Statement} \label{subsec:ProblemStatement}

To minimize the overall latency of data read requests at a frontend server, the number of cached chunks $\varepsilon_m$ for each data item should be optimized as follows~\footnote{For data item $m$, $\varepsilon_m$ data chunks placed at the farthest storage nodes, i.e., with the highest data access latencies, will be cached.}:
\begin{equation}\label{equ:Opt}
	\begin{array}{l}
	\mathop {\min }\limits_{\varepsilon_m \in \mathbb{N}, m \in \mathcal{M}} \sum\limits_{m \in \mathcal{M}} f_m(\varepsilon_m) \cdot r_m \\
	{\rm{s}}.{\rm{t}}.\quad 0 \leq \varepsilon_m \leq K,\\
	\quad \quad \sum\nolimits_{m \in \mathcal{M}} \varepsilon_m \leq C,\\
	\end{array}
\end{equation}
where $r_m$ denotes the user request rate for data item $m$.
Constraint $0 \leq \varepsilon_m \leq K$ ensures that the number of cached chunks for each data is no larger than the number of coded data chunks.
Furthermore, $\sum_{m \in \mathcal{M}} \varepsilon_m \leq C$ ensures that the cache capacity constraint is not violated.
Then, the hardness of the optimization problem~(\ref{equ:Opt}) is examined as follows:
\begin{itemize}
\item Experiments demonstrate that $f_m(\varepsilon_m)$, e.g., the latency function of data item B in Fig.~\ref{fig:Latency_Example}, could be both nonlinear and nonconvex.
    Therefore, problem~(\ref{equ:Opt}) is an integer programming problem with non-convexity and nonlinearity.
    Generally speaking, complex combinatorial techniques are needed for an efficient solution~\cite{Liu_TSC_18}.

\item In a dynamic scenario, the network conditions and the data request rates may be time variant.
    It is a challenge to design online schemes that can react quickly to real-time changes.
\end{itemize}

For a large-scale storage system with uncertainties, the goals of the caching schemes are 1) achieving the ultimate performance of caching on latency reduction, 2) highly efficient for a quick caching decision, and 3) flexible to change the caching decision in an online fashion.

\section{Optimal Caching Schemes Design}\label{sec: Optimal_caching}

In this section, the motivation and design overview are first presented.
Then, assuming that future data popularity and network condition information is available, an offline scheme is designed to find the optimal caching solution.
The results of the optimal scheme can be regarded as a lower bound of the data access latency when the ultimate performance of caching is achieved.

\subsection{Motivation and Design Overview} \label{subsec:Motivation}

As mentioned in Sec.~\ref{subsec:caching}, the latency function $f_m(\varepsilon_m)$ can be different for various data items.
Considering the diversity of chunk storage locations, it is complicated to design a mathematical latency model that is suitable for the entire storage system.
Therefore, the end-to-end latency of data access is used to quantify the benefits of caching.
Through experiments, we analyze the characteristic of data access latencies over WAN.
Based on the deployed experiment platform, the access latencies of data chunk retrieval from remote {\tt buckets} to end users are measured over an hour (from 15:00 to 16:00 on July 30, 2020).
Fig.~\ref{fig:Latency}(a), (b), and (c) show that the data access latencies over WAN remain fairly stable in the long term.
The reason is that the propagation delay dominates and depends primarily on the physical distance of data transmission~\cite{Bogdanov_SoCC_18}.
Experiment results confirm the positive correlation between the physical distance and the latency.
For instance, the data access latency from S\~{a}o Paulo to end users in Victoria is always higher than that from Oregon.

Fig.~\ref{fig:Latency}(d), (e), and (f) demonstrate that the data access latencies are stable for most of the service time.
For example, 89.58\% of the data access latencies from Oregon to Victoria will be in the range of [100, 140] ms.
Besides, 91.94\% of the data access latencies from Ireland to Victoria will be in the range of [650, 700] ms.
For two arbitrary storage nodes, the latency gap also keeps fairly stable in the long term.
The average data access latency can be used to quantify the benefits of caching.
In Sec.~\ref{subsec:OptimalScheme}, assuming that the future data popularity and network condition information are available, an optimal scheme is designed to explore the ultimate performance of caching on latency reduction.

\subsection{Optimal Caching Scheme} \label{subsec:OptimalScheme}

Let $l_i$ denote the average network latency of data access from storage node $i$ to end users for a certain period of time.
According to the storage location information in~(\ref{equ:InitialPlacement}), the average latency of sending data chunk $m_k$ is given by
\begin{equation} \label{equ:Chunk_Latency}
    l_{m_k}= l_i \cdot \mathbf{1}({m_k \rightarrow i}),
\end{equation}
where $\mathbf{1}({m_k \rightarrow i})$ indicates whether data chunk $m_k$ is placed at node $i$ or not, returning 1 if true or 0 otherwise, $k \in \{1, ..., K\}$, $i \in \mathcal{N}$.
For ease of notation, let us relabel the data chunks according to the descending order of the data access latency.
For example, $m_1$ denotes the data chunk placed at the farthest storage node.
Based on the sorted latencies $\{l_{m_1}, ..., l_{m_K}\}$, a $(K+1)$-dimensional array is maintained for each data item
\begin{align}
\begin{split} \label{equ:Latency_Array}
	\{\tau_{m,0}, \tau_{m,1},..., \tau_{m,K}\} = \{0, (l_{m_1}-l_{m_2})\cdot r_m, ... , (l_{m_1}-\\ l_{m_k}) \cdot r_m, ..., (l_{m_1}-l_{m_K}) \cdot r_m, l_{m_1} \cdot r_m\} ,
\end{split}
\end{align}
where $\tau_{m,k-1}=(l_{m_1}-l_{m_k}) \cdot r_m$ represents the value of reduced latency when $k-1$ data chunks are cached.
For example, if chunk $m_1$ and $m_2$ are cached, then $l_{m_3}$ becomes the bottleneck.
Clearly, $\tau_{m,0} = 0$ as no latency will be reduced without caching.
When all $K$ data chunks are cached, the maximum value of reduced latency is $\tau_{m,K} = l_{m_1} \cdot r_m$.
Based on the reduced latency information, an $M \times (K+1)$ valuation array $\tau$ can be maintained for all data items.
%
As shown in Fig.~\ref{fig:Latency_Example}, $f_m(\varepsilon_m)$ is a monotonic decreasing function.
Minimizing the overall data access latency in (\ref{equ:Opt}) is equivalent to maximizing the total amount of reduced latency:
\begin{equation}\label{equ:Opt_maximize}
	\begin{array}{l}
    \mathop {\max } \limits_{\varepsilon_m \in \mathbb{N}, m \in \mathcal{M}} \Theta (\varepsilon_m) = \sum\limits_{m \in \mathcal{M}} \tau_{m, \varepsilon_m} \\
	{\rm{s}}.{\rm{t}}.\quad 0 \leq \varepsilon_m \leq K,\\
	\quad \quad \sum\limits_{m \in \mathcal{M}} \varepsilon_m = C.\\
	\end{array}
\end{equation}

As $C \leq M K$, $\sum_{m \in \mathcal{M}} \varepsilon_m = C$ ensures the cache capacity can be fully utilized for latency reduction.
Then, we determine the optimal decision $\varepsilon_m$ in the following two steps:

\renewcommand{\algorithmiccomment}[1]{\hfill\eqparbox{COMMENT}{$\triangleright$ #1}}
\renewcommand{\algorithmicrequire}{\textbf{Input:}}
\renewcommand{\algorithmicensure}{\textbf{Output:}}
\setcounter{algorithm}{0}
\begin{algorithm}
	\caption{Iterative Search for Cache Partitions}
	\label{Alg:ExhaustiveSearch}
	\begin{algorithmic}[1]
		\REQUIRE Cache capacity $C$, number of coded data chunks $K$, number of data items $M$.
		\ENSURE Set of cache partitions $\chi$.
		\renewcommand{\algorithmicensure}{\textbf{Initialization:}}
		\ENSURE $x_1 \leftarrow C$, $\forall x_k \leftarrow 0$, $k\in\{2,...,K\}$, $\chi \leftarrow \emptyset$.
		\WHILE{$\{x_1,...,x_K\} \notin \chi$}
        \STATE Add $\{x_1,...,x_K\}$ to $\chi$ if $\sum_{k=1}^{K}x_k \leq M$;
        \STATE $x_2 \leftarrow x_2 + 1$ if $x_2 < \hat x_2$ else $x_2 \leftarrow 0$;
        \FOR{$k=3$ to $K$}
        \IF{$x_{k-1}$ is reset to 0}
        \STATE $x_k \leftarrow x_k + 1$ if $x_k < \hat x_k$ else $x_k \leftarrow 0$;
        \ENDIF
        \STATE $x_1=C-\sum_{k=2}^{K}k x_k$;
        \ENDFOR
        \ENDWHILE
	\end{algorithmic}
\end{algorithm}

\subsubsection{{\bf Cache Partitions}}
Let $x_k$ denote the number of data items with $k$ data chunks cached, i.e.,
\begin{equation} \label{equ:Data_num}
	x_k = \sum \nolimits_{m \in \mathcal{M}} \mathbf{1}(\varepsilon_m= k),
\end{equation}
where $\mathbf{1}(\varepsilon_m= k)$ indicates whether $k$ data chunks are cached for data item $m$ or not, $x_k \in \mathbb{N}$.
We define $\{x_1,...,x_K\} \in \chi$ as a potential partition of caching decisions.
Based on the constraints in~(\ref{equ:Opt_maximize}), we have the Diophantine equations as follows
\begin{equation} \label{equ:Diophantine}
\left\{\begin{matrix}
    x_1 + 2x_2 + ... + K x_K = C,
    \\
    x_1 +  x_2 + ... +   x_K \leq M.
\end{matrix}\right.
\end{equation}

All partitions of caching decisions $\chi$ can be derived from~(\ref{equ:Diophantine}) through iterative search.
The pseudo code of the iterative search is listed in Algorithm~\ref{Alg:ExhaustiveSearch}.
Given the value of $\{x_{k+1},...,x_K\}$, the maximum value that $x_k$ can be assigned is
\begin{equation} \label{equ:Max_x_k}
   \hat x_k = \left \lfloor \frac{C-\sum_{n=k+1}^{K}n x_n}{k} \right \rfloor.
\end{equation}

Initially, $\{x_1, x_2, ..., x_K\}=\{C,0,...,0\}$ is a feasible solution if $\sum_{k=1}^{K}x_k = C \leq M$.
We gradually increase the value of $x_2$ from 0.
If $x_2=\hat x_2$, $x_2$ is reset to 0 in the next iteration.
In this way, the value of $x_k$ can be iteratively determined, $k \in \{3,...,K\}$.
If $x_{k-1}=0$, $x_k$ is incremented by 1 next.
If $x_k = \hat x_k$, $x_k$ is also reset to 0 then.
Based on the value of $\{x_{2},...,x_K\}$, $x_1$ is set to $C-\sum_{k=2}^{K}k x_k$.
We repeat the above process until all cache partitions are included in $\chi$.
A simple example is used to demonstrate the iterative search process.
Let $K = 3$ and $C = 5$.
Algorithm~\ref{Alg:ExhaustiveSearch} sequentially appends five cache partitions, i.e., $\{5,0,0\}$, $\{3,1,0\}$, $\{1,2,0\}$, $\{2,0,1\}$, and $\{0,1,1\}$.
The theoretical analysis of the iterative search algorithm is provided as follows.

\begin{property} \label{property:1}
Algorithm~\ref{Alg:ExhaustiveSearch} needs less than $\prod_{k=2}^{K} (\left \lfloor \frac{C}{k}  \right \rfloor + 1)$ iterations to finish.
The size of set $\left | \chi \right |$ is also less than $\prod_{k=2}^{K} (\left \lfloor \frac{C}{k}  \right \rfloor + 1)$.
\end{property}

\begin{proof}
According to~(\ref{equ:Diophantine}), as $\forall x_k \in \mathbb{N}$, the maximum value of $x_k$ is $\left \lfloor \frac{C}{k} \right \rfloor$.
Therefore, $x_k$ can be assigned up to $\left \lfloor \frac{C}{k} \right \rfloor +1$ different values.
As $\{x_{2},...,x_K\}$ cannot be assigned to the maximum values at the same time, Algorithm~\ref{Alg:ExhaustiveSearch} needs less than $\prod_{k=2}^{K} (\left \lfloor \frac{C}{k}  \right \rfloor + 1)$ iterations to obtain all cache partitions in $\chi$.
The set size $\left | \chi \right | < \prod_{k=2}^{K} (\left \lfloor \frac{C}{k}  \right \rfloor + 1)$ also holds.
\end{proof}

\setcounter{algorithm}{1}
\begin{algorithm}
	\caption{Optimal Assignment for Cache Partitions}
	\label{Alg:OptimalCaching}
	\begin{algorithmic}[1]
		\REQUIRE Set of cache partitions $\chi$, valuation array $\tau$, market clearing price $p_m$.
		\ENSURE Caching decision $\varepsilon_m$.
		\renewcommand{\algorithmicensure}{\textbf{Initialization:}}
		\ENSURE $\forall \varepsilon_m,\hat{\varepsilon}_m, p_m \leftarrow 0$.
		\FOR{Cache partition $\{x_1,...,x_K\} \in \chi$}
		\STATE $\mathcal{G} \leftarrow \mathit{preferred\_seller\_graph}(\tau,\{x_1,...,x_K\})$;
		\STATE $\{\mathcal{M}^{[\textup{c}]}, \mathcal{K}^{[\textup{c}]}\} \leftarrow \mathit{constricted\_set}(\mathcal{G})$;
        \STATE $\tau^{\prime} \leftarrow \tau$;  
		\WHILE{$\{\mathcal{M}^{[\textup{c}]}, \mathcal{K}^{[\textup{c}]}\} \neq \emptyset$}
        \FOR{$m \in \mathcal{M}^{[\textup{c}]}$}
        \FOR{$k \in \mathcal{K}^{[\textup{c}]}$}
        \STATE $V_{k} \leftarrow \mathit{sum\_top}(\tau^{\prime}(:,k),x_k)$;
        \STATE $V_{k}^{m} \leftarrow \mathit{sum\_top}(\tau^{\prime}(:,k) \setminus \{\tau^{\prime}_{m,k}\},x_k)$;
        \ENDFOR
        \STATE $p_m \leftarrow p_m + \max\{1, \max\{V_{k}-V_{k}^{m}\}\}$; 
        \STATE $\tau^{\prime}(m,:) \leftarrow \tau(m,:) - p_m$;
        \ENDFOR
        \STATE $\mathcal{G} \leftarrow \mathit{preferred\_seller\_graph}(\tau^{\prime},\{x_1,...,x_K\})$;
		\STATE $\{\mathcal{M}^{[\textup{c}]}, \mathcal{K}^{[\textup{c}]}\} \leftarrow \mathit{constricted\_set}(\mathcal{G})$;
        \ENDWHILE
        \STATE $\hat{\varepsilon}_m \leftarrow k$ according to $\mathcal{G}$;
        \STATE $\forall \varepsilon_m \leftarrow \hat{\varepsilon}_m$ if $\sum_{m \in \mathcal{M}} \tau_{m, \varepsilon_m} < \sum_{m \in \mathcal{M}} \tau_{m, \hat{\varepsilon}_m}$;
        \ENDFOR
	\end{algorithmic}
\end{algorithm}

\subsubsection{{\bf Optimal Assignment for Cache Partitions}}
Recall that each element $x_k$ in a cache partition $\{x_1,..., x_K\}$ represents that $x_k$ data items are selected, each of which cached $k$ chunks in the frontend server.
For example, if data item $m$ is assigned to $x_k$, the caching decision for $m$ becomes $\varepsilon_m=k$.
As shown in Fig.~\ref{fig:Assignment_Example}, the data items and cache partition can be treated as {\bf sellers} and {\bf buyers}, respectively.
According to the valuation array $\tau$, each buyer has a valuation for each data item.
Thus, the optimal assignment can be considered as a market competing for data items with higher valuations.
The pseudo code of the optimal assignment is listed in Algorithm~\ref{Alg:OptimalCaching}.
As shown in Fig.~\ref{fig:Assignment_Example}, buyers may compete for a certain data item.
The basic idea is to increase the price $p_m$ of data item $m$ until the competition is over.
The price $p_m$ is known as {\bf market clearing price}~\cite{Market_Clearing}.
With no competition, the local optimal caching decision $\hat{\varepsilon}_m$ can be obtained for a certain cache partition.
The global optimal assignment is the one that has the maximum valuation among all cache partitions in $\chi$.

\begin{figure}[htbp]
\centerline{\includegraphics[width=3.0in]{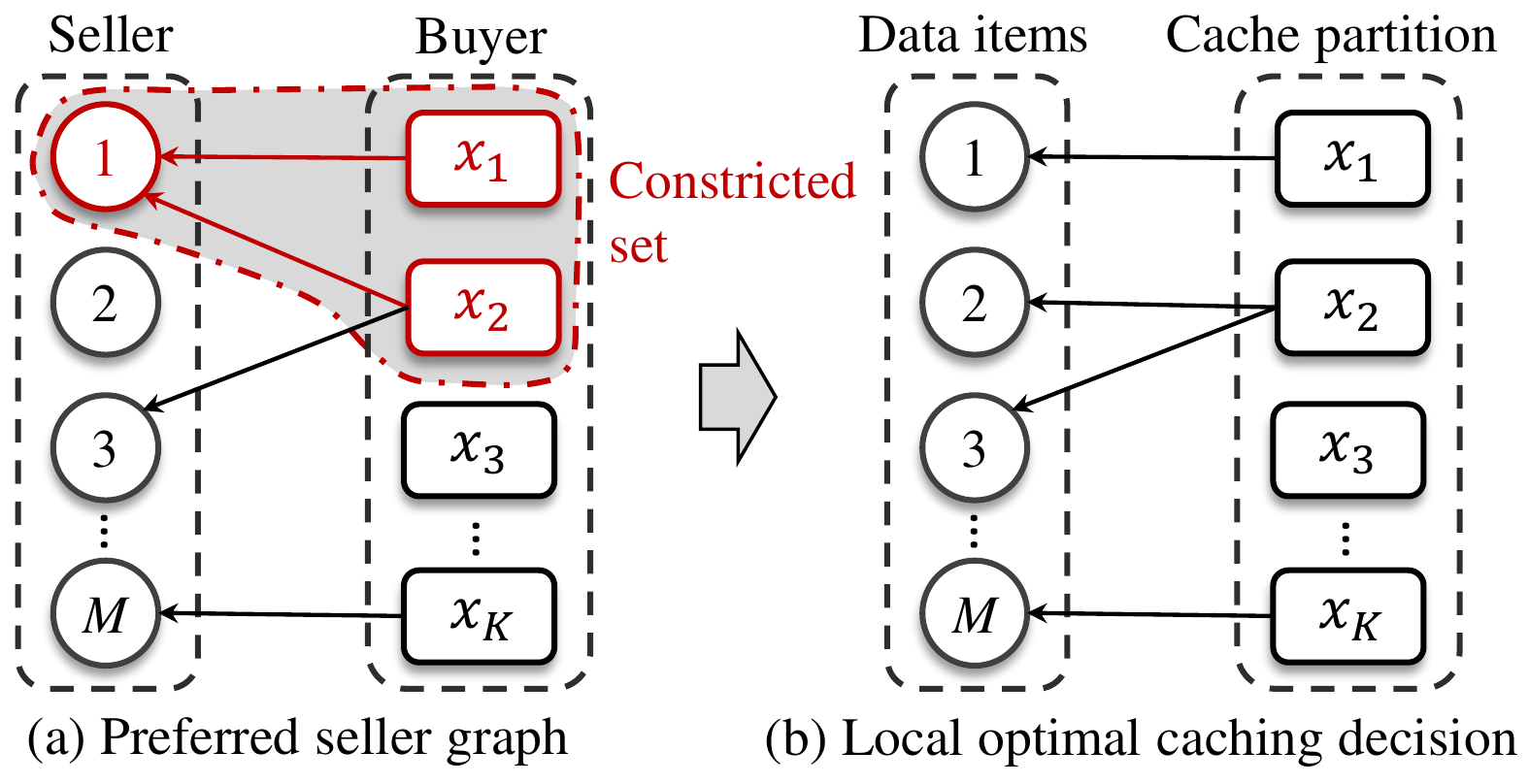}}
\caption{An illustration of the data item assignment for data partition is shown.}
\label{fig:Assignment_Example}
\end{figure}

Let $\tau(:,k)$ denote the $k$-th column of $\tau$, which represents the reduced latencies of all data items when $k$ data chunks of that data item are cached.
To maximize the caching benefits, function $\mathit{preferred\_seller\_graph}$ matches sellers and buyers with the largest $x_k$ elements in $\tau(:,k)$.
As shown in Fig.~\ref{fig:Assignment_Example}(a), a preferred seller graph $\mathcal{G}$ is constructed with function $\mathit{preferred\_seller\_graph}$.
In $\mathcal{G}$, different buyers may compete for the same data item, while each data item can only be assigned to one buyer.
Then, the constricted set $\{\mathcal{M}^{[\textup{c}]}, \mathcal{K}^{[\textup{c}]}\}$ is constructed with function $\mathit{constricted\_set}$, where $\mathcal{M}^{[\textup{c}]}$ denotes the set of competed data, and $\mathcal{K}^{[\textup{c}]}$ represents the set of competing buyers.
Then, we show how to set the market clearing price $p_m$ for each data $m \in \mathcal{M}^{[\textup{c}]}$.
We initialize $p_m=0$, $\forall m \in \mathcal{M}$.
Then, the payoff array $\tau^{\prime}$ can be initialized as the valuation array $\tau$ with $\tau^{\prime} \leftarrow \tau$.
Let $V_{k}$ denote the total payoff of assigning data items (including the competed data item $m$) to buyer $x_k \in \mathcal{K}^{[\textup{c}]}$, i.e.,
\begin{equation} \label{equ:Kvaluations}
    V_{k} = \mathit{sum\_top}(\tau^{\prime}(:,k),x_k),
\end{equation}
where function $\mathit{sum\_top}(\tau^{\prime}(:,k),x_k)$ represents the sum of largest $x_k$ elements in set $\tau^{\prime}(:,k)$.
Then, if data item $m$ is not assigned to $x_k$, the total payoff is given by
\begin{equation} \label{equ:KAvaluations}
    V_{k}^{m} = \mathit{sum\_top}(\tau^{\prime}(:,k) \setminus \{\tau^{\prime}_{m,k}\},x_k).
\end{equation}
If $\max\{V_{k}-V_{k}^{m}\} > 0$ for all buyers in $\mathcal{K}^{[\textup{c}]}$, $p_m$ is incremented by $\max\{V_{k}-V_{k}^{m}\}$.
Then, the payoff of $m$ for all buyers $\tau^{\prime}(m,:)$ is updated as $\tau(m,:) - p_m$.
This ensures $m$ will be assigned to only one buyer $k = \argmax\{V_{k}-V_{k}^{m}\}$ in the next iteration.
If $\max\{V_{k}-V_{k}^{m}\} = 0$, $p_m$ is incremented by the unit price 1.

The above process is repeated until the constricted set is empty.
The whole process needs $M$ iterations at most as at least one data item is excluded from the constricted set in one iteration.
Then, the updated preferred seller graph with no competition is added to the existing assignment.
If the local caching decision $\hat{\varepsilon}_m$ for the current cache partition yields a higher payoff than all previous ones, the global caching decision is updated with $\varepsilon_m \leftarrow \hat{\varepsilon}_m$.
The theoretical analysis of the assignment algorithm is provided as follows.

\begin{theorem}\label{thm:OPT}
	Algorithm~\ref{Alg:OptimalCaching} yields the optimal caching decision.
\end{theorem}

\begin{proof}
    Firstly, we prove that the optimal decision can be obtained for each cache partition.
	This is equivalent to proving that interchanging any two pairs of caching decisions cannot further increase the total valuations.
	Let $m$ and $m'$ denote two randomly selected data items in $\mathcal{G}$.
	With Algorithm~\ref{Alg:OptimalCaching}, let $k$ and $k'$ denote their corresponding number of cached chunks.
	To verify optimality, we need to prove
	\begin{equation} \label{equ:TheoremProof_1}
	    \tau_{m,k}+\tau_{m',k'}\geq \tau_{m',k}+\tau_{m,k'}.
	\end{equation}
	
	If $k$ and $k'$ are not in the constricted set, i.e., $\tau_{m,k}\geq \tau_{m',k}$ and $\tau_{m',k'}\geq \tau_{m,k'}$, we have $\tau_{m,k}+\tau_{m',k'}\geq \tau_{m',k}+\tau_{m,k'}$.
    If $k$ and $k'$ are in the constricted set of $m$ in the previous iteration~\footnote{The case of data item $m'$ can be proved in the same way.} and $m$ is finally assigned to $k$, we have
	\begin{equation} \label{equ:TheoremProof_2}
	    \begin{aligned}
	        V_{k}-V_{k}^{m} \geq V_{k'}-V_{k'}^{m}.	
	    \end{aligned}
	\end{equation}
    Besides, $m'$ is finally assigned to $k'$ with no competition, i.e., $\tau_{m',k}$ is not one of the largest $x_k$ elements in set $\tau^{\prime}(:,k)$.
    We have
    \begin{equation} \label{equ:Final_assign}
    \left\{\begin{matrix}
    V_{k}-V_{k}^{m} \leq \tau_{m,k}-\tau_{m',k},
    \\
    V_{k'}-V_{k'}^{m} = \tau_{m,k'}-\tau_{m',k'}.
    \end{matrix}\right.
    \end{equation}
    This means
	\begin{equation}
	    \begin{aligned}
	      \tau_{m,k}-\tau_{m',k} \geq \tau_{m,k'}-\tau_{m',k'},\\
	    \end{aligned}
	\end{equation}
    which concludes that the optimal caching decision is obtained for the cache partition.
	As all partitions in $\chi$ are considered, Algorithm~\ref{Alg:OptimalCaching} yields the global optimal caching decision.
\end{proof}

\begin{property}\label{pro:Computation_complexity}
The computation complexity of Algorithm~\ref{Alg:OptimalCaching} is less than $O(\frac{C^{K-1} \cdot M^2}{(K-1)!})$.
\end{property}

\begin{proof}
To obtain the preferred seller graph, all columns in $\tau$ are sorted via the radix sort algorithm.
The sorting complexity is $O(M K)$ (Line 2).
Then, all data items need to be considered to determine the constricted set with the complexity of $O(M)$ (Line 3).
As all buyers may compete for a data item, the calculation of the market clearing price needs $K$ iterations at most (Line 7--12).
Furthermore, the preferred seller graph and the constricted set are updated with the complexity of $O(M + M K)$ (Line 14--15).
As discussed above, the while loop needs $M$ iterations at most.
Furthermore, the data item assignment in Line 17 and 18 needs $K$ iterations at most.
The optimal assignment for a cache partition needs $M^2(K+1) + 2MK + M + K$ iterations at most.
The computation complexity for a cache partition is $O(M^2 K)$.
Considering all cache partitions in $\chi$, the computation complexity of Algorithm~\ref{Alg:OptimalCaching} is less than $O(\frac{C^{K-1} \cdot M^2}{(K-1)!})$.
\end{proof}

Property~\ref{pro:Computation_complexity} demonstrates that the computation complexity of Algorithm~\ref{Alg:OptimalCaching} is mainly determined by the total number of cache partitions $\left | \chi \right |$.
Based on the experiment platform deployed in Sec.~\ref{subsec:caching}, Table~\ref{tab:Impact_of_C} in Sec.~\ref{subsec:Factors} shows the number of cache partitions $\left | \chi \right |$ and the average running time (ART) of Algorithm~\ref{Alg:OptimalCaching} under different settings.
The number of required iterations $\left | \chi \right |$ and the ART increase rapidly with the increase of cache capacity $C$ and the number of coded data chunks $K$.
This means that Algorithm~\ref{Alg:OptimalCaching} may incur a heavy computation burden for a large-scale storage system.
Furthermore, the long running time implies that the optimal scheme cannot react quickly to real-time network changes.
The network states may change before caching decisions can be updated.
To sum up, the optimal scheme is an offline solution with the requirement of future data popularity and network condition information.

\section{Online Caching Scheme Design} \label{sec: Online_caching}

Guided by the optimal caching scheme in Sec.~\ref{subsec:OptimalScheme}, an online caching scheme is proposed with no assumption about future data popularity and network condition information.
Furthermore, we extend the proposed caching schemes to the case of storage server failure.

\subsection{Online Caching Scheme}

Let $\mathcal{T}$ denote the whole data service period.
The online scheme updates the caching decision according to the measured data popularity $r_m^t$ and network latencies $l_i^t$ in real time, $t \in \mathcal{T}$.

{\bf Data Popularity}: The Discounting Rate Estimator (DRE)~\cite{DRE} method is applied to construct the real-time request information $r_m^t$.
On the frontend server, a counter is maintained for each data item, which increases with every data read, and decreases periodically with a ratio factor.
The benefits of DRE are as follows: 1) it reacts quickly to the request rate changes, and 2) it only requires $O(1)$ space and update time to maintain the prediction for each counter.

{\bf Network Latency}: Similar to~\cite{Liu_DataBot_19}, the Exponentially Weighted Moving Average (EWMA) method~\cite{EWMA} is used to estimate the average network latency of data requests.
Specifically, after a data read operation, $l_i^t$ is updated by
\begin{equation}
	l_i^t  = \alpha_l \cdot l_i^t + (1-\alpha_l) \cdot \iota _i,
\end{equation}
where $\iota _i$ is the measured end-to-end latency of a data request, and $\alpha_l$ is the discount factor to reduce the impact of previous requests.
The advantage of EWMA is that it only needs $O(1)$ space to maintain the prediction for each storage node.

Let $\Gamma$ denote the set of data requests in the service period $\mathcal{T}$.
To ensure the adaptivity of our design, the caching decision is updated upon the arrival of each request $\gamma_m^t$, $\gamma_m^t \in \Gamma$, $t \in \mathcal{T}$.
To improve the solution efficiency, there is no need to modify the storage system by completely overriding the existing caching decisions.
The valuation $\{\tau_{m,0},\tau_{m,1},..., \tau_{m,K}\}$ is updated according to the latest measurement of data access latency $l_i^t$ and request rate $r_m^t$.

\setcounter{algorithm}{2}
\begin{algorithm}
	\caption{Online Caching Scheme}
	\label{Alg:OnlineCaching}
	\begin{algorithmic}[1]
		\REQUIRE Cache capacity $C$, number of coded data chunks $K$, number of data items $M$, valuation array $\tau$, set of data requests $\Gamma$ in period $\mathcal{T}$.
		\ENSURE Set of cached data items $\hat{\mathcal{M}}$, online caching decision $\varepsilon_m^t$, $m \in \mathcal{M}$.
		\renewcommand{\algorithmicensure}{\textbf{Initialization:}}
		\ENSURE $\hat{\mathcal{M}} \leftarrow \emptyset$, $\forall \varepsilon_m^t \leftarrow 0$. 
        \FOR {Data request $\gamma_m^t \in \Gamma$, $t \in \mathcal{T}$}
        \STATE Update $\{\tau_{m,0},\tau_{m,1},..., \tau_{m,K}\}$ according to~(\ref{equ:Chunk_Latency}) and~(\ref{equ:Latency_Array});
        \IF {$\sum_{n \in \mathcal{M}} \varepsilon_n^t \leq C-K$ and $\varepsilon_m^t < K$}
        \STATE $\varepsilon_m^t \leftarrow K$, add $m$ to $\hat{\mathcal{M}}$;
        \ELSIF {$\sum_{n \in \mathcal{M}} \varepsilon_n^t > C-K$ and $\varepsilon_m^t < K$}
        \STATE $\hat{\mathcal{M}}^{\prime} \leftarrow \{m\}$;
        \REPEAT
        \STATE $n \leftarrow \argmin_{n \in \hat{\mathcal{M}} \setminus \hat{\mathcal{M}}^{\prime}}\{\frac{\tau_{n,k}}{k}\}$, add $n$ to $\hat{\mathcal{M}}^{\prime}$;
        \UNTIL $K \leq C - \sum_{n \in \hat{\mathcal{M}}} \varepsilon_n^t + \sum_{n' \in  \hat{\mathcal{M}}^{\prime}} \varepsilon_{n'}^t \leq 2K-1$
        \STATE $\forall \varepsilon_{n'}^t \leftarrow 0$, $n' \in \hat{\mathcal{M}}^{\prime}$;
        \STATE Invoke Algorithm~\ref{Alg:ExhaustiveSearch} for cache partition set $\hat{\chi}$ based on the available cache capacity;
        \STATE Invoke Algorithm~\ref{Alg:OptimalCaching} to update the caching decisions $\varepsilon_{n'}^t$ based on $\hat{\mathcal{M}}^{\prime}$ and $\hat{\chi}$, $n' \in \hat{\mathcal{M}}^{\prime}$;
        \STATE Update $\hat{\mathcal{M}}$, remove $n$ from $\hat{\mathcal{M}}$ if $\varepsilon_n^t=0$, $\forall n \in \hat{\mathcal{M}}^{\prime}$;
        \ENDIF
        \ENDFOR
	\end{algorithmic}
\end{algorithm}

Let $\hat{\mathcal{M}}$ denote the set of already cached data items.
If the cache capacity is not fully utilized, i.e., $\sum_{n \in \hat{\mathcal{M}}} \varepsilon_n^t \leq C-K$, all $K$ data chunks of the requested data item $m$ should be cached for latency reduction.
In contrast, if $\sum_{n \in \hat{\mathcal{M}}} \varepsilon_n^t > C-K$, we need to determine 1) whether data item $m$ should be cached or not, 2) how many chunks for $m$ should be cached, and 3) which data items in $\hat{\mathcal{M}}$ should be replaced?
To solve this problem, the data items in $\hat{\mathcal{M}}$ with the lowest valuations per unit are added into subset $\hat{\mathcal{M}}^{\prime}$.
The data items in $\hat{\mathcal{M}}^{\prime}$ are expected to be replaced first by the requested data item $m$ to maximize the total amount of reduced latency.
Furthermore, $m$ is also added into $\hat{\mathcal{M}}^{\prime}$.
All data items in $\hat{\mathcal{M}}^{\prime}$ are cache replacement candidates.
The cached data items in $\hat{\mathcal{M}}$ are gradually added into $\hat{\mathcal{M}}^{\prime}$ until the available cache capacity $C - \sum_{n \in \hat{\mathcal{M}}} \varepsilon_n^t + \sum_{n' \in  \hat{\mathcal{M}}^{\prime}} \varepsilon_{n'}^t \geq K$.
This guarantees that all $K$ data chunks of $m$ have a chance to be cached.
The expansion of $\hat{\mathcal{M}}^{\prime}$ needs $K$ iterations at most with $\left |\hat{\mathcal{M}}^{\prime} \right | \leq K+1$ and $C - \sum_{n \in \hat{\mathcal{M}}} \varepsilon_n^t + \sum_{n' \in  \hat{\mathcal{M}}^{\prime}} \varepsilon_{n'}^t \leq 2K-1$.
Based on the available cache capacity, Algorithm~\ref{Alg:ExhaustiveSearch} is invoked to calculate the cache partition set $\hat{\chi}$.
Then, based on subset $\hat{\mathcal{M}}^{\prime}$ and $\hat{\chi}$, Algorithm~\ref{Alg:OptimalCaching} is invoked to update the caching decisions $\varepsilon_n^t$, $n \in \hat{\mathcal{M}}^{\prime}$.
The pseudo code of the online caching scheme is listed in Algorithm~\ref{Alg:OnlineCaching}.
The theoretical analysis of the online scheme is provided as follows.

\begin{theorem}\label{thm:Online}
	Algorithm~\ref{Alg:OnlineCaching} yields the worst-case approximation ratio of $1-\frac{2K-1}{C}$ upon the arrival of data requests.
\end{theorem}

\begin{proof}
In Algorithm~\ref{Alg:OnlineCaching}, the greedy selection of $\hat{\mathcal{M}}^{\prime}$ (Line 8) may incur performance loss.
Let $\varepsilon_m^t = k$ denote the caching decision obtained with Algorithm~\ref{Alg:OnlineCaching} for request $\gamma_m^t$, $0 \leq k \leq K$.
Then, we consider the following two different cases:

1) $\sum_{n \in  \hat{\mathcal{M}}^{\prime} \setminus \{m\}} \tau_{n, \varepsilon_n^t} \leq \tau_{m,K}$: Since Algorithm~\ref{Alg:OptimalCaching} is invoked, Algorithm~\ref{Alg:OnlineCaching} outputs the optimal decision for subset $\hat{\mathcal{M}}^{\prime}$.
As $\tau_{m,k} \leq \tau_{m,K}$, the obtained objective value $\Theta$ satisfies
\begin{equation} \label{equ:Theorem2Proof_1}
    \Theta \geq \sum \nolimits_{n \in  \hat{\mathcal{M}}} \tau_{n, \varepsilon_n^t} - \sum \nolimits_{n \in  \hat{\mathcal{M}}^{\prime} \setminus \{m\}} \tau_{n, \varepsilon_n^t} + \tau_{m,K}.
\end{equation}

The global optimal objective value satisfies
\begin{equation} \label{equ:Theorem2Proof_2}
    \Theta^* \leq \sum \nolimits_{n \in  \hat{\mathcal{M}}} \tau_{n, \varepsilon_n^t} + \tau_{m,K}.
\end{equation}

Due to the greedy selection, $\frac{\tau_{n', \varepsilon_{n'}^t}}{\varepsilon_{n'}^t} \leq \frac{\tau_{n, \varepsilon_n^t}}{\varepsilon_n^t}$ holds, $\forall n' \in \hat{\mathcal{M}}^{\prime}$, $\forall n \in \hat{\mathcal{M}} \setminus \hat{\mathcal{M}}^{\prime}$.
As $\sum_{n \in  \hat{\mathcal{M}}^{\prime}} \varepsilon_n^t \leq 2K-1$, we have
\begin{equation} \label{equ:Theorem2Proof_3}
    \frac{\sum_{n \in  \hat{\mathcal{M}}} \tau_{n, \varepsilon_n^t}}{C} \geq \frac{\sum_{n \in  \hat{\mathcal{M}}^{\prime} \setminus \{m\}} \tau_{n, \varepsilon_n^t}}{2K-1}.
\end{equation}

The worst-case performance bound is given by
\begin{align}
\begin{split}
    \frac{\Theta}{\Theta^*}& \geq \frac{C-2K+1}{C}.
\end{split}
\end{align}

2) $\sum_{n \in  \hat{\mathcal{M}}^{\prime} \setminus \{m\}} \tau_{n, \varepsilon_n^t} > \tau_{m,K}$: In this case, we have
\begin{align}
\begin{split}
    \Theta^*& < \sum \nolimits_{n \in  \hat{\mathcal{M}}} \tau_{n, \varepsilon_n^t} + \tau_{m,K} \\
    & < \sum \nolimits_{n \in  \hat{\mathcal{M}}} \tau_{n, \varepsilon_n^t}+\sum \nolimits_{n \in  \hat{\mathcal{M}}^{\prime} \setminus \{m\}} \tau_{n, \varepsilon_n^t}.
\end{split}
\end{align}

As $\Theta \geq \sum_{n \in  \hat{\mathcal{M}}} \tau_{n, \varepsilon_n^t}$, we have
\begin{align}
\begin{split}
    \frac{\Theta}{\Theta^*}& > \frac{\sum_{n \in  \hat{\mathcal{M}}} \tau_{n, \varepsilon_n^t}}{\sum_{n \in  \hat{\mathcal{M}}} \tau_{n, \varepsilon_n^t}+\sum_{n \in  \hat{\mathcal{M}}^{\prime} \setminus \{m\}} \tau_{n, \varepsilon_n^t} } > \frac{C}{C+2K-1}.
\end{split}
\end{align}
The proof completes.
\end{proof}

\begin{property}\label{pro:Computation_complexity_2}
For a data request, the computation complexity of Algorithm~\ref{Alg:OnlineCaching} is less than $O((K+1)^3 \cdot K!)$.
\end{property}

\begin{proof}

For any data requests, we have $\left |\hat{\mathcal{M}}^{\prime} \right | \leq K+1$ and $C - \sum_{n \in \hat{\mathcal{M}}} \varepsilon_n^t + \sum_{n' \in  \hat{\mathcal{M}}^{\prime}} \varepsilon_{n'}^t \leq 2K-1$.
According to Property~\ref{property:1}, the set size $\left | \hat{\chi} \right | < \prod_{k=2}^{K} (\left \lfloor \frac{2K-1}{k}  \right \rfloor + 1)$ holds.
As $\frac{2K-1}{k} < K-k+2$, $k \in \{2,...,K\}$, we have $\left \lfloor \frac{2K-1}{k}  \right \rfloor + 1 \leq K-k+2$.
Therefore, $\left | \hat{\chi} \right | < K!$ holds.
Furthermore, similar to the analysis in Property~\ref{pro:Computation_complexity}, the computation complexity for a cache partition is $O((K+1)^3)$.
The computation complexity of Algorithm~\ref{Alg:OnlineCaching} is less than $O((K+1)^3 \cdot K!)$.
\end{proof}

In large-scale storage systems, the number of coded data chunks per data item $K$ is much smaller than the cache capacity $C$ and the number of data items $M$, i.e., $K \ll C$ and $K \ll M$.
Therefore, Theorem~\ref{thm:Online} shows that Algorithm~\ref{Alg:OnlineCaching} can approximate the optimal solution well.
Furthermore, Property~\ref{pro:Computation_complexity_2} shows that the computation complexity of Algorithm~\ref{Alg:OnlineCaching} is tremendously reduced when compared with that of Algorithm~\ref{Alg:OptimalCaching}.
Table~\ref{tab:Impact_of_C} in Sec.~\ref{subsec:Factors} shows the maximum number of required iterations $\left | \hat{\chi} \right |$ and the ART of Algorithm~\ref{Alg:OnlineCaching}.
The low computation complexity ensures that the online scheme can react quickly to real-time changes.

\subsection{Caching with Server Failure} \label{subsec:EC-Caching-With-Failure}

The proposed optimal and online schemes work well without server failure.
However, servers may experience downtime in the distributed storage system.
In this subsection, the proposed caching schemes are extended to the case of storage server failure.
Let $\mathcal{M}_i$ denote the set of remotely unavailable data chunks when a storage server at node $i$ fails.
If data chunk $m_k \in \mathcal{M}_i$ is not cached beforehand, the degraded read is triggered to serve the data requests.
The parity chunk $m_r$ with the lowest data access latency will be fetched from node $j$ for data reconstruction.
The unavailable data chunk $m_k$ is replaced by parity chunk $m_r$, i.e., $m_k \leftarrow m_r$ and $l^t_{mk} \leftarrow l^t_{mr}$.
Similar to~(\ref{equ:Chunk_Latency}), the average latency of sending $m_r$ is given by
\begin{equation} \label{equ:Parity_Chunk_Latency}
	l^t_{m_r} = \min \{l^t_j \cdot \mathbf{1}({m_r \rightarrow j})\}.
\end{equation}

When Algorithm~\ref{Alg:OptimalCaching} or~\ref{Alg:OnlineCaching} suggest caching $m_r$, the recovered data chunk $m_k$ (instead of $m_r$) is directly added into the caching layer.
In this way, the decoding overheads of the subsequent data requests can be reduced.
This means our design still works well when server failure happens.

\begin{figure*}[!t]
	\centering
	\begin{minipage}[b]{0.315\textwidth}
		\centering
		\includegraphics[width=2.35in]{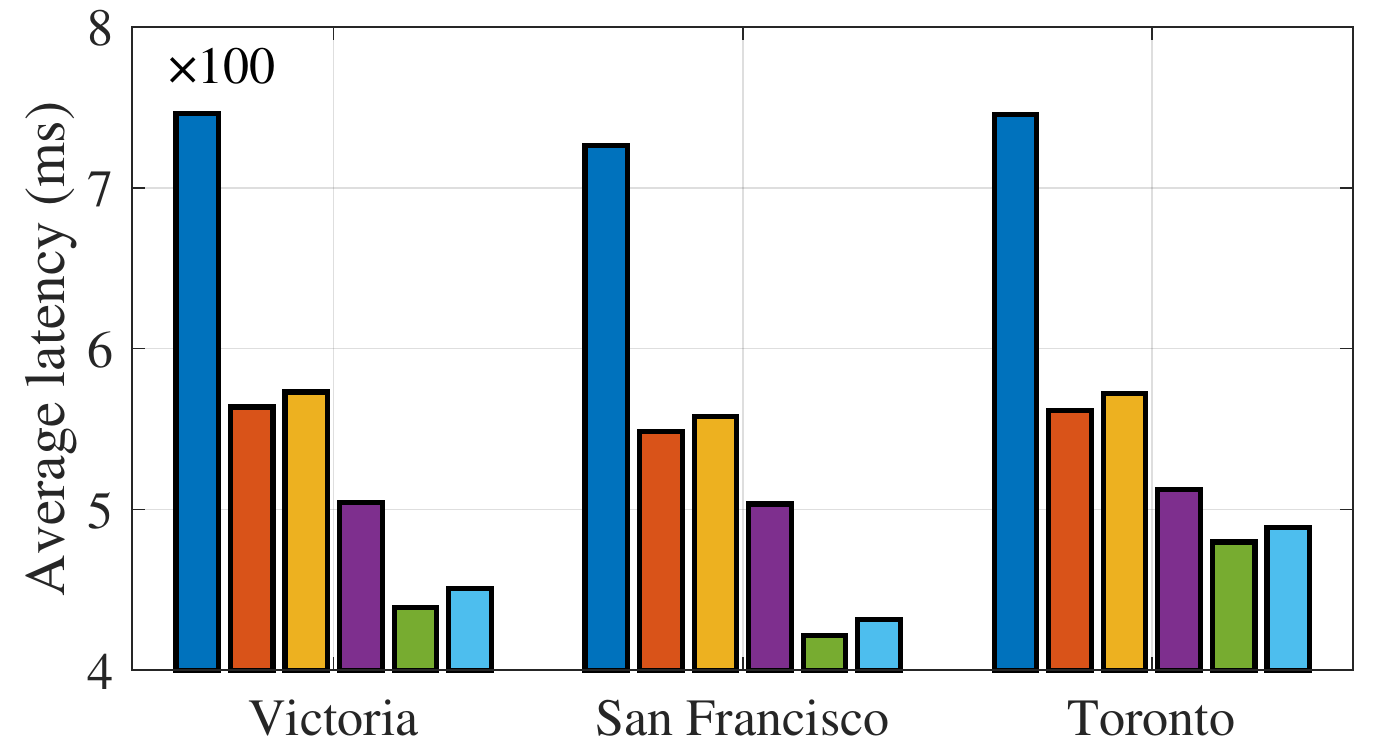}
		\caption{Average data request latencies.}
		\label{fig:Average_latency_case}
	\end{minipage}
	\hspace{5pt}
	\begin{minipage}[b]{0.315\textwidth}
		\centering
		\includegraphics[width=2.35in]{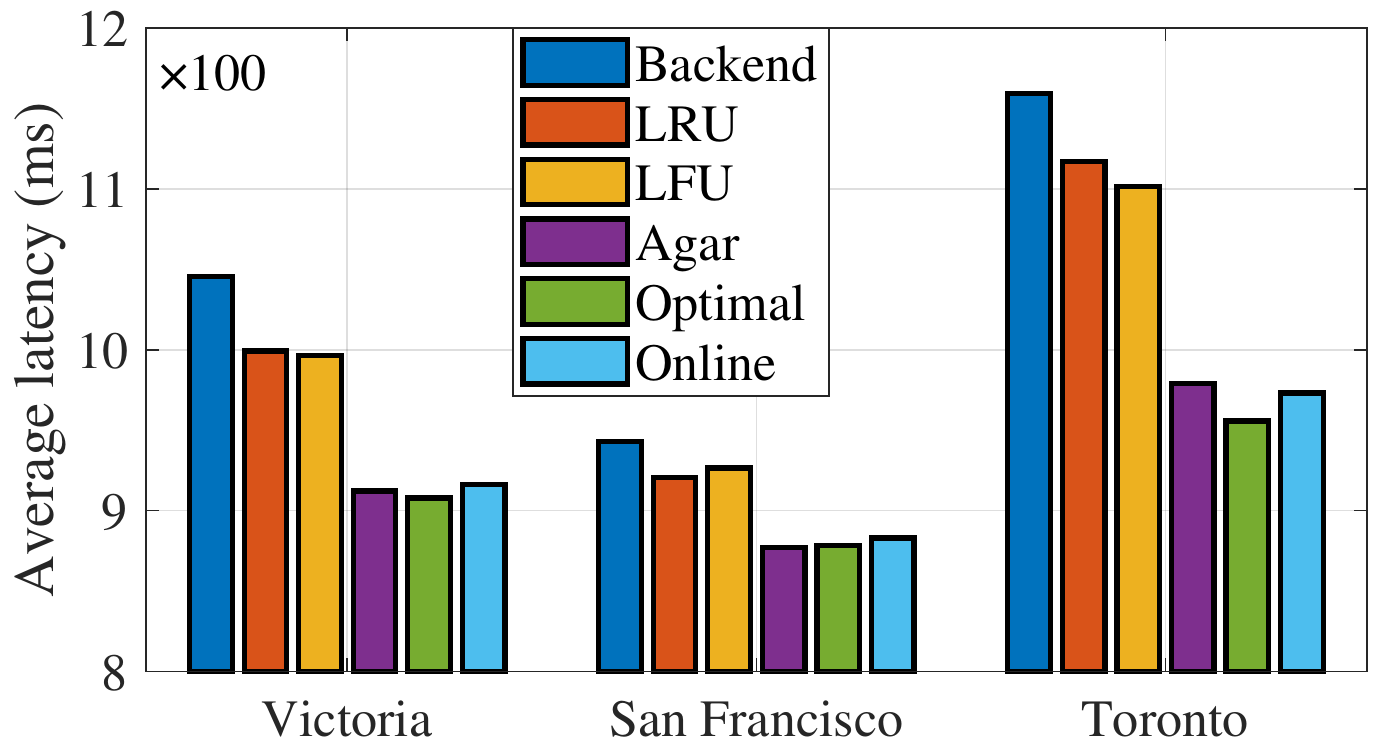}
		\caption{95th percentile tail latencies.}
		\label{fig:tail_latencies_case}
	\end{minipage}%
    \hspace{5pt}
	\begin{minipage}[b]{0.315\textwidth}
		\centering
		\includegraphics[width=2.35in]{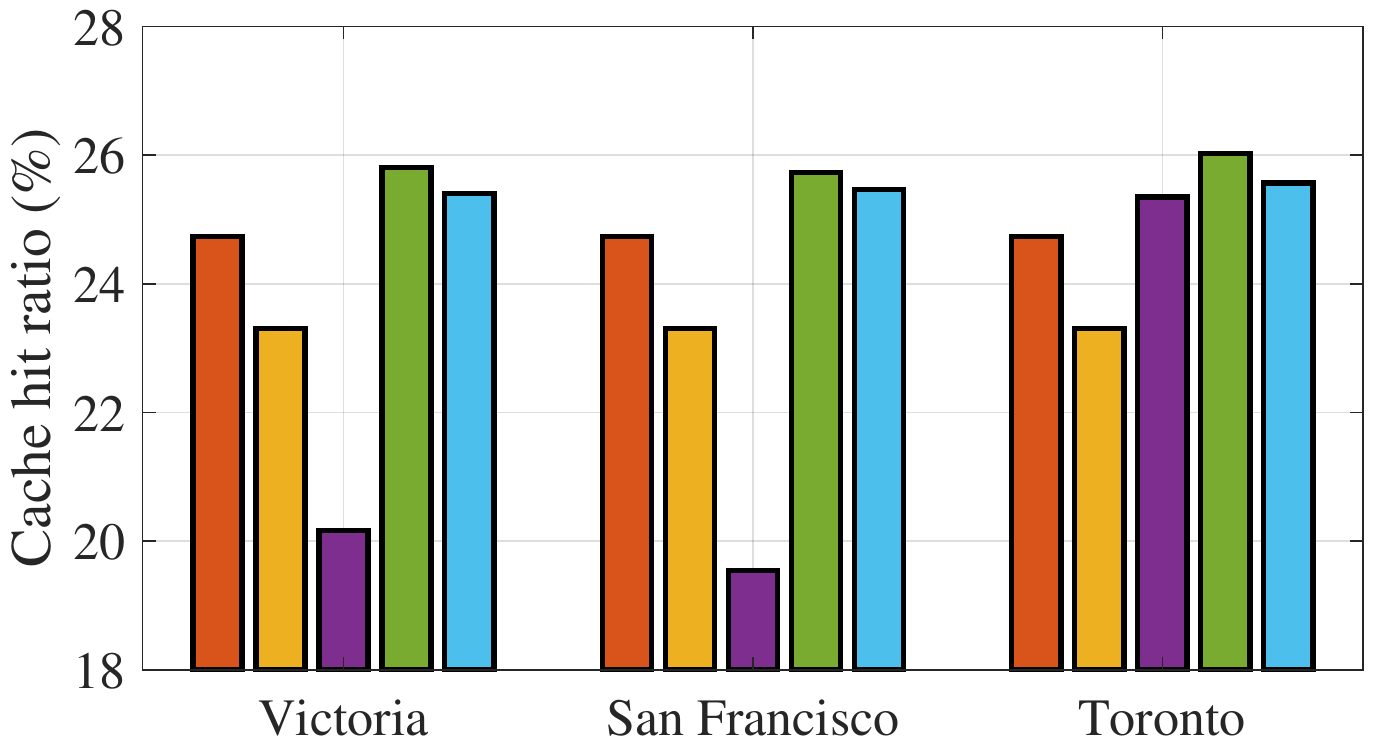}
		\caption{Hit ratio of data chunk requests.}
		\label{fig:Hit_ratio}
	\end{minipage}%
\end{figure*}

\section{Experimental Evaluation} \label{sec:Evaluation}

In this section, we build a prototype of the caching system in Python and then integrate it with the experiment platform deployed in Sec.~\ref{subsec:caching}.
Next, extensive experiments are performed to quantitatively evaluate the performance of the proposed optimal and online caching schemes.

\textbf{Experimental Setup:}
We deploy eighteen {\tt buckets} in $N = 6$ AWS regions.
Each {\tt bucket} denotes a remote storage server.
The library {\tt zfec}~\cite{zfec} is adopted to implement the RS codes.
By default, we set the number of coded chunks $K = 6$ and $R = 3$.
The coded chunks of $M=1,000$ data items are with the same size 1 MB~\cite{Agar_17}.
They are uniformly distributed at different {\tt buckets} to achieve fault tolerance.
As shown in Table~\ref{tab:AWS_Regions}, three frontend servers are built on personal computers in different Cities.
The hardware features an Intel(R) Core(TM) i7-7700 HQ processor and 16 GB memory.
The cache capacity of {\tt Memcached} is set to 100 MB in RAM, i.e., the maximum number of cached data chunks is set to $C = 100$.
The data service period $\mathcal{T}$ is set to 1 hour.
Similar to the previous studies~\cite{ECCache_16,Hu_SoCC17}, the popularity of data requests follows a Zipf distribution, which is common in many real-world data request distributions.
The tail index of the Zipf distribution is set to 2.0 under the default settings, i.e., highly skewed.

\textbf{Performance Baselines:} For a fair performance comparison, four other schemes are adopted as performance baselines.

\begin{itemize}

\item  Backend---All required $K$ data chunks are directly fetched from the remote {\tt buckets} with no caching.
This scheme is adopted to quantify the benefits of caching.

\item  LRU and LFU---The Least Recently Used (LRU) and Least Frequently Used (LFU) caching policies are used to replace the contents in the caching layer.
For each selected data item, all $K$ data chunks are cached.

\item Agar~\cite{Agar_17}---A dynamic programming-based scheme is designed to iteratively add data chunks with larger request rates and higher data access latencies in the caching layer.
\end{itemize}

\subsection{Experimental Results} \label{subsec:Results}

To begin with, the performances of six schemes, i.e., Backend, LRU, LFU, Agar, and the proposed optimal and online schemes, are compared under the default settings.
As all requested data chunks are fetched from the remote {\tt buckets}, Backend yields the highest average latency (746.27 ms, 726.43 ms, and 745.34 ms) and 95th percentile tail latencies (1045.49 ms, 943.12 ms, and 1159.16 ms) at three frontend servers.

LRU caches the recently requested data items by discarding the least recently used data items.
LFU caches the data items with higher request rates.
Compared with Backend, LRU and LFU reduce the average latencies of all data requests at three frontend servers by 24.6\% and 23.2\%, respectively.
As illustrated in Fig.~\ref{fig:Hit_ratio}, 24.7\% and 23.3\% of requests are fulfilled by the cached data chunks with LRU and LFU.
With the whole data items cached, LRU and LFU reduce the access latencies to 0 ms for 24.7\% and 23.3\% of data requests, respectively.
However, as the cache capacity is limited, the rest parts of the requests suffer high access latencies.
Compared with Backend, the 95th percentile tail latencies are only reduced by 3.5\% and 3.9\%, respectively.
LFU and LRU overlook the diversity of data chunk storage locations and the heterogeneity of latencies across different storage nodes.
Caching the whole data items cannot enjoy the full benefits of caching.

Agar iteratively improves the existing caching configurations by considering new data chunks.
Each data item is assigned with a weight, given by the number of data chunks to cache.
Compared with LFU and LRU, more data items can enjoy the benefits of caching.
The average latencies at three frontend servers are reduced to 504.21 ms, 503.42 ms, and 512.27 ms, respectively.
Moreover, Agar prefers to evict low valuation data chunks which incur high access latencies from the caching layer.
The 95th percentile tail latencies are reduced to 912.24 ms, 877.02 ms, and 978.97 ms.

With an overall consideration of data request rates and access latencies, the proposed optimal scheme optimizes the number of cached data chunks for each data item, minimizing the average latencies to 439.06 ms, 421.52 ms, and 479.79 ms.
Fig.~\ref{fig:Average_latency_case} shows that the proposed online scheme approximates the optimal scheme well with a similar latency of 450.95 ms, 431.66 ms, and 488.84 ms.
As shown in Table~\ref{tab:Impact_of_C}, although raising the average latency by 2.3\%, the online scheme greatly reduces the computation overheads.
Furthermore, the proposed optimal and online schemes optimize the caching decisions by selecting the data chunks with higher valuations.
This means the contents in the caching layers are with higher data request rates and lower access latencies.
The hit ratios of data requests from three frontend servers are 25.8\%, 25.7\%, and 26.0\%, respectively.
The 95th percentile tail latencies with the optimal scheme are reduced to 907.86 ms, 878.11 ms, and 955.65 ms.
The online scheme incurs a similar 95th percentile tail latencies of 916.45 ms, 883.0 ms, and 973.13 ms.

\subsection{Impact of Other Factors} \label{subsec:Factors}

\begin{figure*}[!t]
	\centering
	\begin{minipage}[b]{0.315\textwidth}
		\centering
		\includegraphics[width=2.3in]{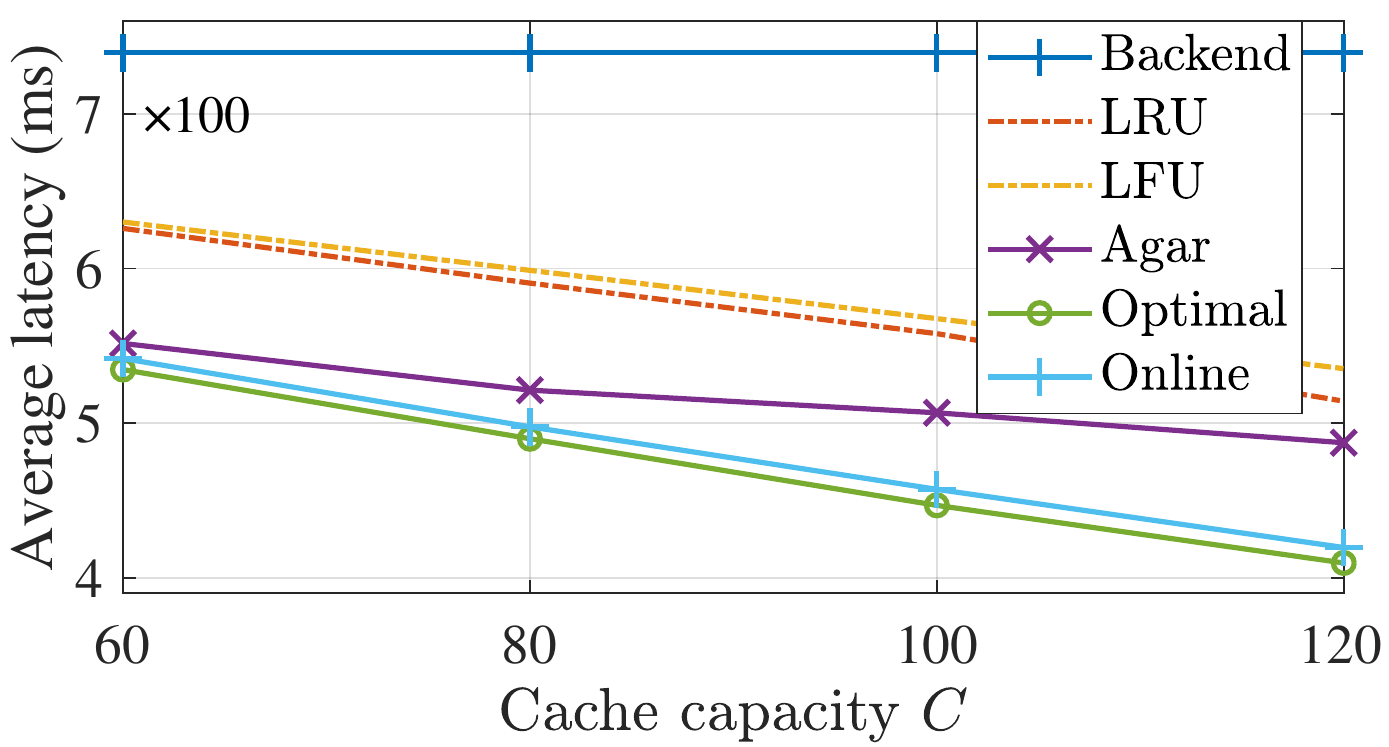}
		\caption{Impact of cache capacity.}
		\label{fig:Cache_capacity}
	\end{minipage}
	\hspace{5pt}
	\begin{minipage}[b]{0.315\textwidth}
		\centering
		\includegraphics[width=2.3in]{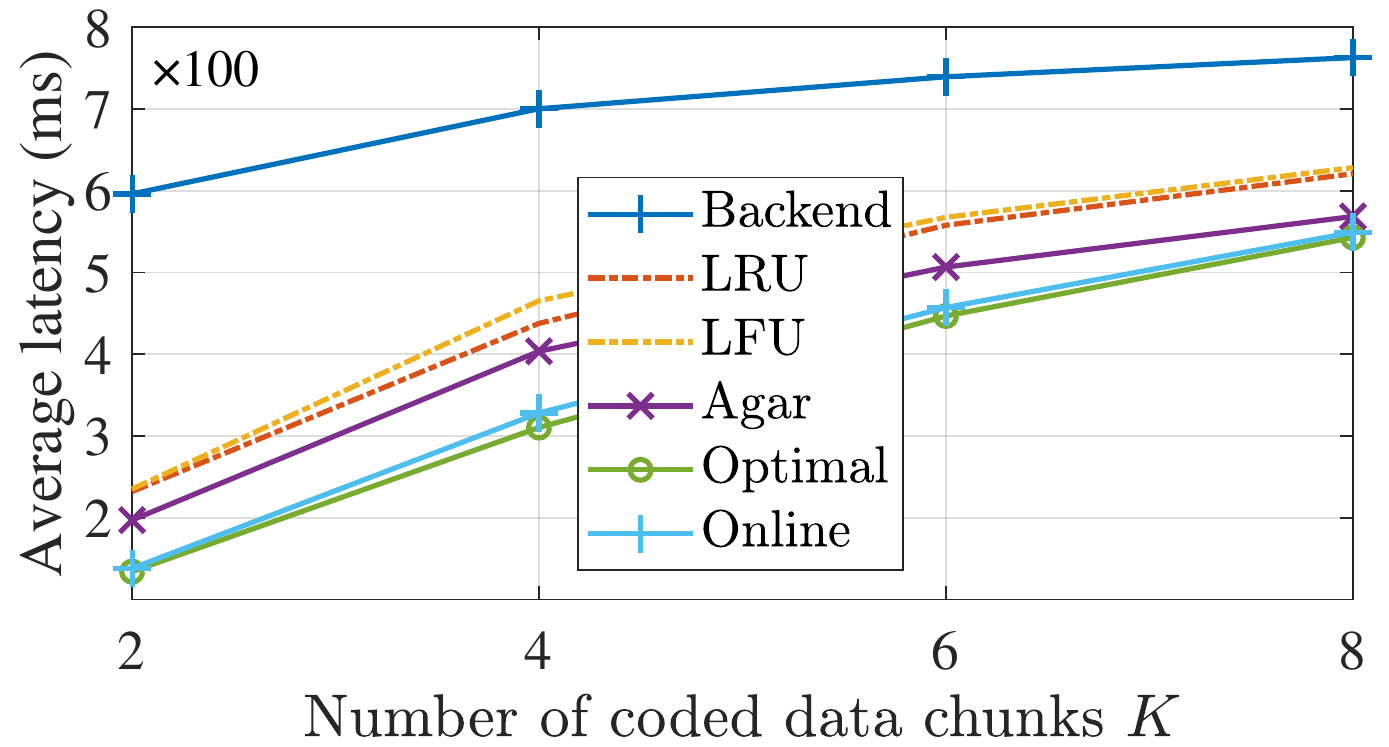}
		\caption{Impact of number of coded data chunks.}
		\label{fig:EC_num}
	\end{minipage}%
    \hspace{5pt}
	\begin{minipage}[b]{0.315\textwidth}
		\centering
		\includegraphics[width=2.3in]{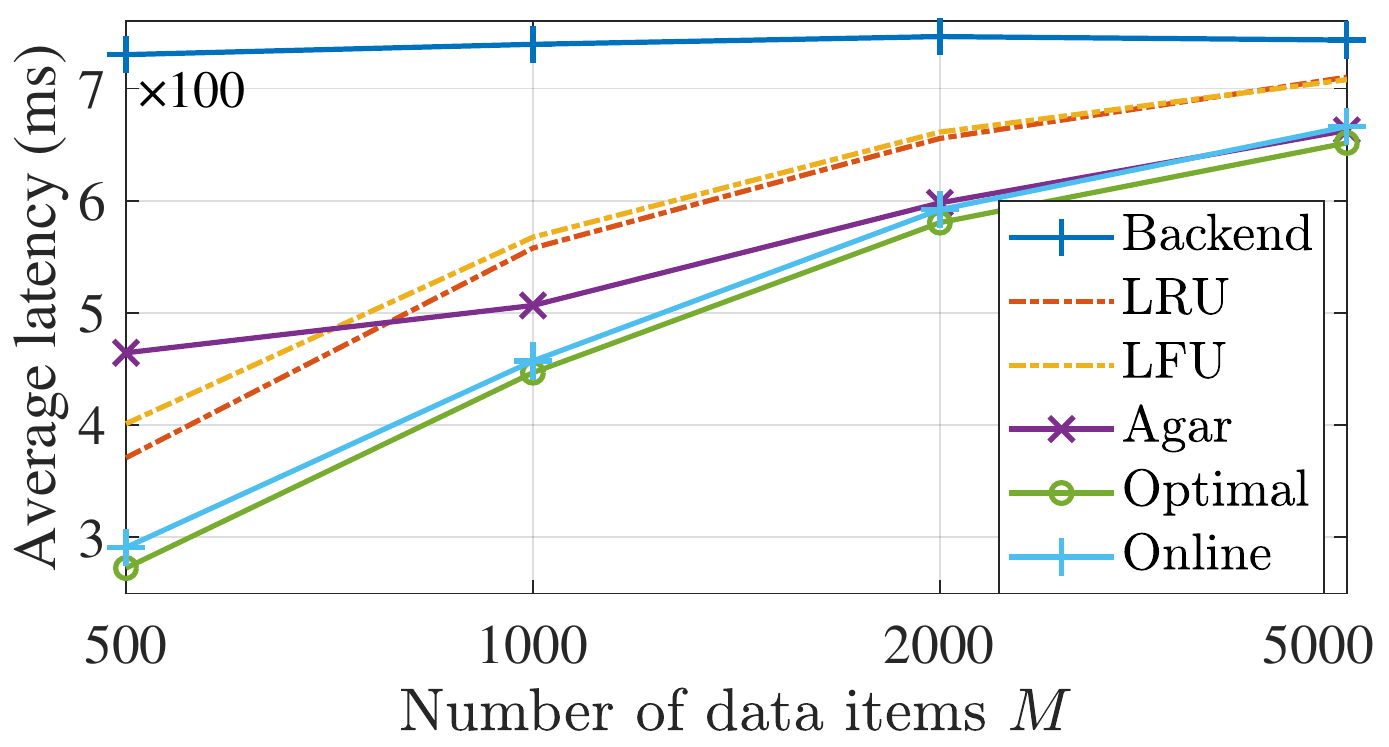}
		\caption{Impact of number of data items.}
		\label{fig:Impact_M}
	\end{minipage}%
\end{figure*}

\begin{table}[!t]
\caption{The number of cache partitions and the ART of caching schemes with the variation of cache capacity and number of coded data chunks.}
\begin{center}
\begin{tabular}{|c|c|c|c|c|c|}
\hline
\multicolumn{2}{|c|}{\textbf{$C$}} & \textbf{60} & \textbf{80} & \textbf{100} & \textbf{120}\\
\hline
\textbf{LRU} & {ART (ms)} & \multicolumn{4}{|c|}{0.05}\\
\hline
\textbf{LFU} & {ART (ms)} & \multicolumn{4}{|c|}{0.05}\\
\hline
\textbf{Agar} & {ART (ms)} & 544.48 & 661.37 & 800.02 & 960.97\\
\hline
\multirow{2}*{\textbf{Optimal}} & $\left | \chi \right |$ & 19,858 & 69,624 & 189,509 & 436,140\\
\cline{2-6}
& {ART (s)} & 391.64 & 2239.67 & 7863.99 & 23746.93\\
\hline
\multirow{2}*{\textbf{Online}} & $\max \left | \hat{\chi} \right |$ & \multicolumn{4}{|c|}{37} \\
\cline{2-6}
& {ART (ms)} & 2.29 & 1.93 & 1.71 & 1.55\\
\hline
\hline
\multicolumn{2}{|c|}{\textbf{$K$}} & \textbf{2} & \textbf{4} & \textbf{6} & \textbf{8} \\
\hline
\textbf{LRU} & {ART (ms)} & \multicolumn{4}{|c|}{0.05}\\
\hline
\textbf{LFU} & {ART (ms)} & \multicolumn{4}{|c|}{0.05}\\
\hline
\textbf{Agar} & {ART (ms)} & 421.74 & 628.31 & 800.02 & 948.22\\
\hline
\multirow{2}*{\textbf{Optimal}} & $\left | \chi \right |$ & 51 & 8,037 & 189,509 & 1,527,675 \\
\cline{2-6}
& ART (s) & 2.04 & 666.27 & 7863.99 & 49662.91 \\
\hline
\multirow{2}*{\textbf{Online}} & $\max \left | \hat{\chi} \right |$ & 2 & 9 & 37 & 127 \\
\cline{2-6}
& ART (ms) & 0.28 & 0.90 & 1.71 & 6.22 \\
\hline
\end{tabular}
\label{tab:Impact_of_C}
\end{center}
\end{table}

In this section, the impacts of cache capacity, number of coded data chunks, number of data items, data popularity, and server failure, are considered for performance evaluation.
For simplicity, the average latency represents the average latency of all data requests from three frontend servers in the following of the paper.

\textbf{Cache Capacity:} Fig.~\ref{fig:Cache_capacity} illustrates the average latency when the cache capacity $C$ increases from 60 to 120 chunks.
With no data caching, the average latencies of Backend remain stable at 739.35 ms.
With more data requests enjoy the caching benefits, the average latencies with all five caching schemes decrease.
With the increase of $C$, the proposed schemes have more space for caching decision optimization.
Compared with Agar, the percentage of reduced latency via the proposed optimal scheme is improved from 3.1\% to 16.0\%.
As shown in Table~\ref{tab:Loss}, when compared with the optimal scheme, the online scheme only increases the average latency from 1.3\% to 2.4\% with the variation of cache capacity.

Then, the ART of five caching schemes is evaluated, which determines the efficiency of deploying a caching solution.
As shown in Table~\ref{tab:Impact_of_C}, by using simple heuristics, LRU and LFU only need 0.05 ms to update the caching decision.
Agar periodically optimizes the caching configuration for all data items in the storage system, which needs hundreds of milliseconds for a round of optimization.
With the increase of cache capacity, the number of cache partitions $\left | \chi \right |$ increases rapidly from  19,858 to 436,140.
The ART of the optimal scheme increases from 391.64 s to 23746.93 s.
In contrast, the online scheme updates the caching decision upon the arrival of each data request.
According to our design in Algorithm~\ref{Alg:OnlineCaching}, the maximum number of cache partitions $\left | \hat{\chi} \right |$ is determined by the number of coded data chunks $K$.
When a data request arrives, the caching decision will not be updated if the data item is already cached.
Therefore, the ART of the online scheme for each request decreases from 2.29 ms to 1.55 ms with the increase of cache capacity.
This means the online scheme is a scalable solution for a large-scale storage system.

\textbf{Number of Coded Data Chunks:} The size of data items is increased from 2 MB to 8 MB.
With the same size of coded chunks (1 MB), the number of coded data chunks $K$ increases from 2 to 8.
As coded chunks are uniformly distributed among remote {\tt buckets}, more data chunks will be placed at the {\tt buckets} with higher access latencies with the increase of $K$.
Therefore, as shown in Fig.~\ref{fig:EC_num}, with $C=100$ and $M=1,000$, the average data access latency with Backend increases from 596.06 ms to 762.83 ms.
Moreover, when the data item is coded into more data chunks, more requests are served by fetching data chunks from the remote {\tt buckets}.
The average latencies with all five caching schemes increase accordingly.
Fig.~\ref{fig:EC_num} shows that the proposed optimal and online schemes always incur lower latencies than Agar, LRU, LFU, and Backend.
Compared with Agar, the percentage of reduced latency via the online scheme varies from 29.9\% to 3.4\% with the increase of $K$.
Furthermore, Table~\ref{tab:Impact_of_C} shows that the ART of the online scheme only increases from 0.28 ms to 6.22 ms.
With the online scheme, few extra delays will be introduced to handle the intensive data requests.

\begin{table}[!t]
\caption{The percentage of increased data access latency incurred by the online scheme, i.e., performance loss, when compared with the optimal scheme.}
\begin{center}
\begin{tabular}{|c|c|c|c|c|c|}
\hline
\textbf{$C$} & \textbf{60} & \textbf{80} & \textbf{100} & \textbf{120}\\
\hline
\textbf{Performance loss (\%)} & 1.3  &  1.6  &  2.3 &  2.4 \\
\hline
\hline
\textbf{$K$} & \textbf{2} & \textbf{4} & \textbf{6} & \textbf{8}\\
\hline
\textbf{Performance loss (\%)} & 3.0 &   5.7  &  2.3  &  1.2 \\
\hline
\hline
\textbf{$M$} & \textbf{500} & \textbf{1,000} & \textbf{2,000} & \textbf{5,000}\\
\hline
\textbf{Performance loss (\%)} & 6.7 & 2.3 & 2.0 & 2.2\\
\hline
\hline
\textbf{Tail index} & \textbf{Uniform} & \textbf{1.2} & \textbf{1.5} & \textbf{2.0}\\
\hline
\textbf{Performance loss (\%)} & 2.2  &  2.2  &  1.7  &  2.3\\
\hline
\end{tabular}
\label{tab:Loss}
\end{center}
\end{table}

\begin{figure*}[!t]
	\centering
	\begin{minipage}[b]{0.315\textwidth}
		\centering
		\includegraphics[width=2.3in]{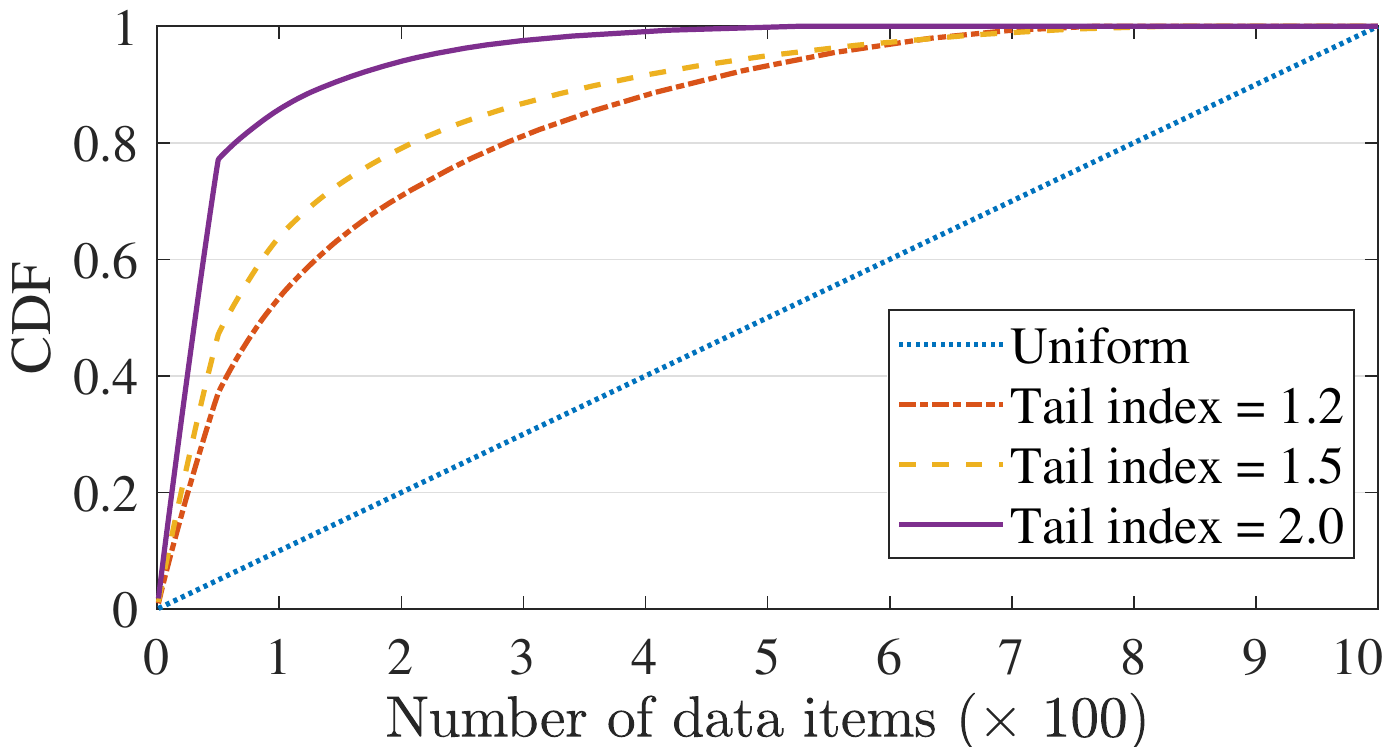}
		\caption{CDF of data popularity.}
		\label{fig:Skewness_CDF}
	\end{minipage}
	\hspace{5pt}
	\begin{minipage}[b]{0.315\textwidth}
		\centering
		\includegraphics[width=2.3in]{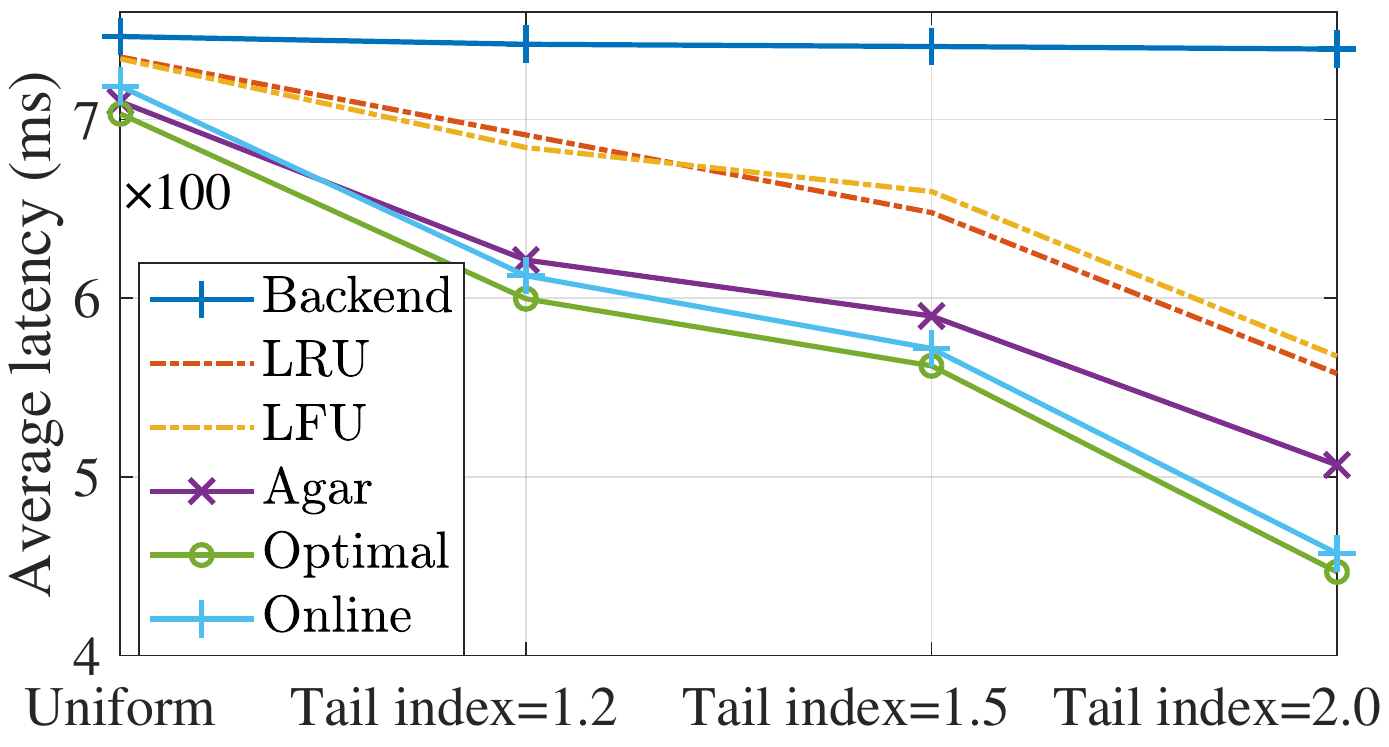}
		\caption{Impact of data popularity.}
		\label{fig:Skewness}
	\end{minipage}%
    \hspace{5pt}
	\begin{minipage}[b]{0.315\textwidth}
		\centering
		\includegraphics[width=2.3in]{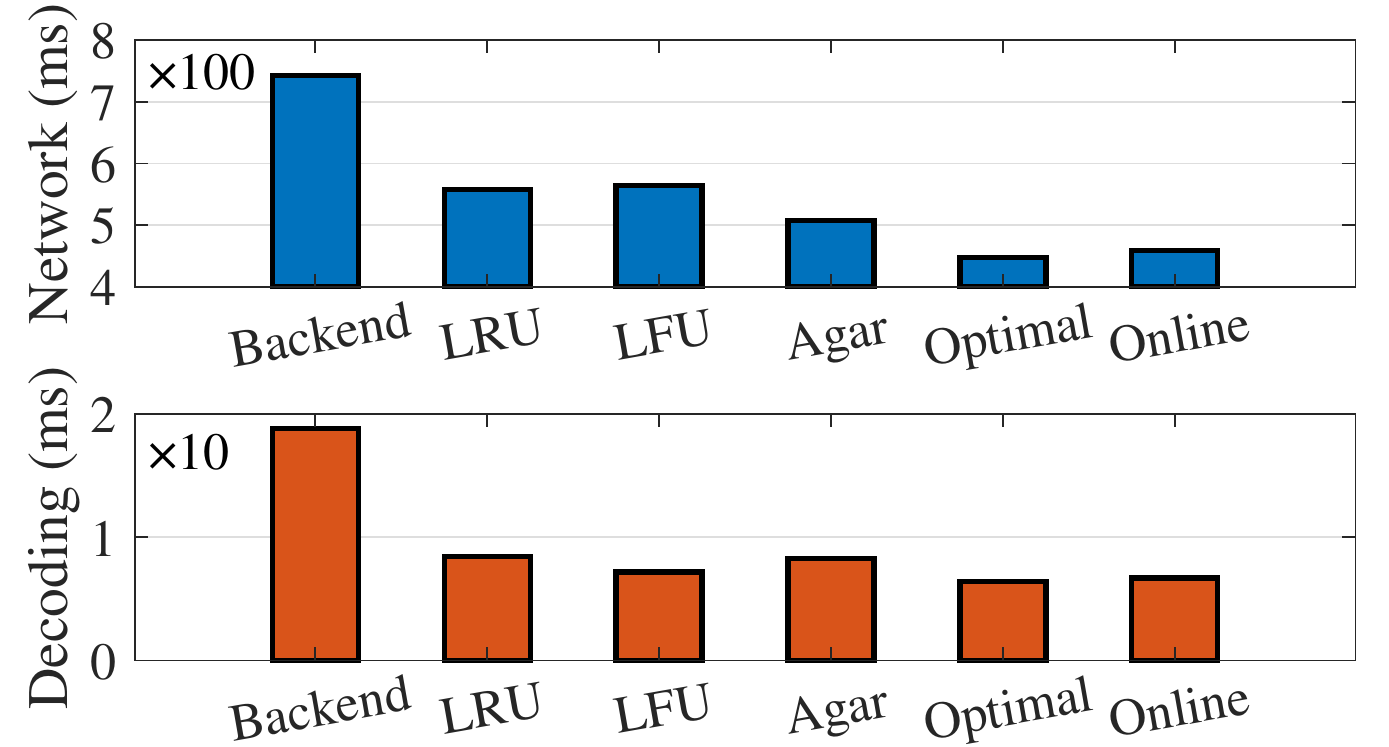}
		\caption{Impact of server failure.}
		\label{fig:Failure}
	\end{minipage}%
\end{figure*}

\textbf{Number of Data Items:} As shown in Fig.~\ref{fig:Impact_M}, with $C=100$ and $K=6$, the number of deployed data items $M$ is increased from 500 to 5,000.
The average data access latency with Backend remains basically the same.
As the data popularity follows Zipf distribution, a small portion of data items get the majority of data requests.
With the growing total number of data items, the number of data items with relatively higher request rates increases.
Due to the limited cache capacity, more and more data requests are served by fetching data chunks from the remote servers.
Therefore, the average latencies increase rapidly with LRU (from 370.93 ms to 710.07 ms), LFU (from 401.26 ms to 707.84 ms), Agar (from 464.45 ms to 662.62 ms), and the optimal (from 272.66 ms to 651.48 ms) and online (from 291.01 ms to 666.0 ms) schemes.

\textbf{Data Popularity:} Fig.~\ref{fig:Skewness_CDF} illustrates the CDF of the data popularity using uniform and Zipf distributions.
As shown in Fig.~\ref{fig:Skewness}, all six schemes incur similar data access latencies when the data popularity follows a uniform distribution.
When all data items are with the same popularity, the caching valuation is only determined by the storage locations of data items.
With a similar caching valuation for different data items, the benefits of caching are not significant when the cache capacity is limited.
Then, with the increase of tail index from 1.2 to 2.0, the skew of the data popularity becomes higher and higher.
A fraction of data items with higher request frequencies can benefit more from caching.
When the tail index is set to 2.0, compared with Backend, LRU, LFU, and Agar, the optimal scheme reduces the average latency by 39.6\%, 21.3\%, 19.9\%, and 11.8\%, respectively.
Furthermore, Table~\ref{tab:Loss} demonstrates that the online scheme can approximate the optimal scheme well.
When compared with the optimal scheme, the online scheme only increases the average latency by about 2\% under different settings of data popularity.

\textbf{Server Failure:} Then, we evaluate the performance of the proposed optimal and online schemes when server failure happens.
For a fair performance comparison, LRU, LFU, and Agar also cache the recovered data chunks (instead of parity chunks) to reduce the decoding overheads.
Please note that erasure codes can tolerate up to $R$ simultaneous server failures.
Recent research indicated that single server failure is responsible for 99.75\% of all kinds of server failures~\cite{Khan_FAST_12}.
Therefore, single server failure is considered in this paper by terminating each storage server in turn.
The experiment setting is identical to that in Sec.~\ref{subsec:Results} except for storage server failure.
If the needed data chunks are not available on the remote servers or cached in the caching layer, degraded read will be triggered to serve the data requests.
In this case, the data access latency contains two parts, i.e., the network latency and the decoding latency.

Fig.~\ref{fig:Failure} illustrates the average data access latencies with various schemes.
Without caching services, Backend incurs the average network latency of 742.81 ms and the average decoding latency of 18.82 ms.
By caching the recovered data chunks to avoid unnecessary decoding overheads of the subsequent data requests, LFU, LRU, Agar, and the proposed optimal and online schemes reduce the average decoding latencies by more than 55\%.
Compared with Backend, LRU, LFU, and Agar, the optimal scheme reduces the overall average data access latency by 40.4\%, 20.5\%, 19.7\%, and 12.0\%, respectively.
Furthermore, compared with the optimal scheme, the online scheme incurs a performance loss of 2.5\% in the presence of server failure.

\section{Conclusion and Future Work} \label{sec:conclusion}

In this paper, novel caching schemes were proposed to achieve low latency in the distributed coded storage system.
To reduce the data access latency, frontend servers, each with an in-memory caching layer, were deployed to cache coded data chunks near end users.
Experiments based on Amazon S3 confirmed the positive correlation between the latency and the physical distance of data retrieval over the WAN.
As the distributed storage system spans multiple geographical sites, the average data access latency was used to quantify the benefits of caching.
With the assumption of future data popularity and network latency information, an optimal caching scheme was proposed to obtain the lower bound of data access latency.
Guided by the optimal scheme, we further designed an online caching scheme based on the measured data popularity and network latencies in real time.
Extensive experiments demonstrated that the online scheme approximates the optimal scheme well and significantly reduces the computation complexity.
In future work, more performance metrics, e.g., load balance among storage nodes, will be considered in the caching problem.

%
%
%

\ifCLASSOPTIONcaptionsoff
  \newpage
\fi

\end{document}